\documentclass[10pt,british]{article}
\pdfoutput=1

\usepackage[a4paper]{geometry}

\usepackage[T1]{fontenc}
\usepackage[utf8]{inputenc}
\usepackage[final]{microtype}

\usepackage{abstract}

\newcommand*{\email}[1]{\href{mailto:#1}{\tt #1}}

\usepackage{babel}
\usepackage{csquotes}
\usepackage[all]{foreign}
\defnotforeign[wrt]{{\xperiodafter{w.r.t}}}

\usepackage{amsthm}
\usepackage{amsmath}
\usepackage{amsfonts,amssymb,bm,stmaryrd}
\usepackage{mathtools}
\usepackage{esvect}

\allowdisplaybreaks

\usepackage{thmtools}

\usepackage{hyperref}
\usepackage[capitalise,noabbrev,nameinlink]{cleveref}

\hypersetup{
	hidelinks,
	final,
}

\usepackage{cite}
\theoremstyle{plain}
\newtheorem{theorem}{Theorem}[section]

\newtheorem{proposition}[theorem]{Proposition}

\theoremstyle{definition}
\newtheorem{definition}{Definition}[section]

\theoremstyle{remark}
\newtheorem{remark}{Remark}[section]
\newtheorem{example}[remark]{Example}

\usepackage{subcaption}
\usepackage[noend,tworuled,linesnumbered]{algorithm2e}
\usepackage[section]{placeins}

\usepackage[inline]{enumitem}

\newenvironment{conditions}[3][]{
		\begin{enumerate}[
			label={\upshape\bfseries(#2\arabic*)},
			ref={\upshape #2\arabic*},
			align=left, labelindent=\parindent,leftmargin=\parindent,
			#1
		]
		\let\olditem\item
		\renewcommand*\item[1][#3\arabic{enumi}]{\olditem\label[condition]{##1}}
	}{\end{enumerate}}

\crefname{condition}{Condition}{Conditions}

\usepackage{tikz}

\usetikzlibrary{backgrounds}
\usetikzlibrary{fit}
\usetikzlibrary{calc}
\usetikzlibrary{arrows}
\usetikzlibrary{matrix}
\usetikzlibrary{positioning}
\usetikzlibrary{intersections}
\usetikzlibrary{decorations}
\usetikzlibrary{automata}
\usetikzlibrary{through}
\usetikzlibrary{mindmap}
\usetikzlibrary{shapes.geometric}

\tikzstyle{ssplace}=[circle,thick,draw=black,fill=black!5,minimum size=6.5mm, inner sep=0.8pt]
\tikzstyle{blackdot}=[circle, draw=black, fill=black]
\tikzstyle{blackline}=[draw=black, thick, minimum size=6mm]
\tikzstyle{blackarrow}=[->,blackline]
\tikzstyle{transplace}=[circle,thick,draw=black!15,fill=black!5,minimum size=6mm]
\tikzstyle{simularrow}=[very thick, black!70, dashed]
\tikzstyle{etaarrow}=[very thick, black!70, dashed]
\tikzstyle{snaky}=[decorate, decoration={snake, segment length=7mm, amplitude=0.3mm}]

\tikzstyle{supplace}=[circle, very thick,draw=black,fill=black!5,minimum size=10mm, inner sep=2pt]
\tikzstyle{minssplace}=[circle,thick,draw=black,fill=black!5,minimum size=5mm, inner sep=0pt]

\tikzstyle{process}=[circle, very thick,draw=black,fill=black!5,minimum size=10mm, inner sep=2pt]

\tikzstyle{labelsq}=[thick, draw=black, fill=black!5, minimum size=6mm, inner sep=2pt, rectangle,
minimum width={width("Monitora")+2pt}]

\tikzstyle{handler}=[thick, draw=black, fill=black, minimum size=4mm, inner sep=0pt, rectangle]

\tikzset{
 position/.style args={#1:#2 from #3}{
	 at=(#3.#1), anchor=#1+180, shift=(#1:#2)
 }
}

\usepackage{framed}

\makeatletter
\let\orgdescriptionlabel\descriptionlabel
\renewcommand*{\descriptionlabel}[1]{%
	\let\orglabel\label
	\let\label\@gobble
	\phantomsection
	\edef\@currentlabel{#1}%
	\let\label\orglabel
	\orgdescriptionlabel{#1}%
}
\makeatother

\makeatletter
\newcommand{\Superimpose}[2]{%
	{\ooalign{$#1\@firstoftwo#2$\cr\hfil$#1\@secondoftwo#2$\hfil\cr}}}
\makeatother

\makeatletter
\newcommand{\footnoteref}[1]{%
	\protected@xdef\@thefnmark{\ref{#1}}\@footnotemark%
}
\makeatother

\newcommand{\Attr}{\mathcal{A}}
\newcommand{\Must}{\mathcal{M}} 
\newcommand{\Unique}{\mathcal{U}} 
\newcommand{\Injective}{\mathcal{E}} 
\newcommand{\Choice}{\mathcal{C}} 
\newcommand{\MustSym}{\text{$\exists$}} 
\newcommand{\UniqueSym}{\text{$!$}} 
\newcommand{\InjectiveSym}{\text{\reflectbox{\rotatebox[origin=c]{180}{$!$}}}} 

\newcommand{\Pwset}{\mathcal{P}} 
\newcommand{\Path}{\mathbb{P}} 
\newcommand{\SSD}[2]{\mathbb{S}^{#1\to#2}} 
\newcommand{\R}{\mathcal{R}} 
\newcommand{\N}{\mathbb{N}} 
\newcommand{\Lang}{\mathcal{L}} 
\newcommand{\injarrow}{\hookrightarrow} 
\newcommand{\BO}{\mathcal{O}}
\newcommand{\Until}{\mathbin{\mathcal{U}\kern-.1em}}

\newcommand{\out}{\mathrm{out}}

\newcommand{\fakemiddlepipe}[2]{{\kern-\nulldelimiterspace\left.\vphantom{{#1}{#2}}\right|}}

\makeatletter
\newcommand{\my@arrow}[1]{\ooalign{$#1-\mkern-5mu-$\cr\hidewidth$#1>$}}
\newcommand{\pointerarrow}{\mathrel{\mathpalette\my@arrow\relax}}
\makeatother

\hypersetup{
	pdftitle={Loose graph simulations},
	pdfauthor={Alessio Mansutti, Marino Miculan and Marco Peressotti}
}

\title{Loose Graph Simulations}

\author{
	Alessio Mansutti
	\\\small
	University of Udine
	\\\small
	\email{alessio.mansutti@lsv.fr}
	\and
	Marino Miculan\thanks{Partially supported by PRID 2017 \emph{ENCASE} of the University of Udine.}
	\\\small
	University of Udine
	\\\small
	\email{marino.miculan@uniud.it}
	\and
	Marco Peressotti\thanks{Partially supported by the Open Data Framework project at the University of Southern Denmark, and by the Independent Research Fund Denmark, Natural Sciences, grant no.~DFF-7014-00041.}
	\\\small
	University of Southern Denmark
	\\\small
	\email{peressotti@imada.sdu.dk}
}

\date{}

\begin{document}

\maketitle

\begin{abstract}
We introduce \emph{loose graph simulations} (LGS), a new notion about labelled graphs which subsumes in an intuitive and natural way \emph{subgraph isomorphism} (SGI), \emph{regular language pattern matching} (RLPM) and \emph{graph simulation} (GS).
Being a unification of all these notions, LGS allows us to express directly also problems which are ``mixed'' instances of previous ones, and hence which would not fit easily in any of them.
After the definition and some examples, we show that the problem of finding loose graph simulations is NP-complete, we provide formal translation of SGI, RLPM, and GS into LGSs, and we give the representation of a  problem which extends both SGI and RLPM. Finally, we identify a subclass of the LGS problem that is polynomial.
\end{abstract}

\section{Introduction}\label{sec:INTRODUCTION}
\emph{Graph pattern matching} is the problem of finding patterns satisfying a specific property, inside a given graph.
This problem arises naturally in many research fields: for instance, in computer science it is used in automatic system verification, network analysis and data mining
\cite{DBLP:conf/sigcomm/LischkaK09,DBLP:conf/gg/1997handbook,DBLP:conf/icdm/YanH02,Chakrabarti2006GraphML};
in computational biology it is applied to protein sequencing  \cite{DBLP:books/daglib/0002478}; in cheminformatics it is used to study molecular systems and predict their evolution \cite{DBLP:journals/bmcbi/BonniciGPSF13,bioinflmu-308}.
As a consequence, many definitions of patterns have been proposed;
for instance, these patterns can be specified by another graph, by a formal language, by a logical predicate, etc.
This situation has led to different notions of graph pattern matching, such as \emph{subgraph isomorphism} (SGI), \emph{regular language pattern matching} (RLPM) and \emph{graph simulation} (GS). Each of these notions has been studied in depth, yielding similar but different theories, algorithms and tools.

A drawback of this situation is that it is difficult to deal with matching problems which do not fit directly in any of these variants.
In fact, often we need to search for patterns that can be expressed as compositions of several graph pattern matching notions.
An example is when we have to find a pattern which has to satisfy multiple notions of graph pattern matching at once; due to the lack of proper tools, these notions can only be checked one by one with a worsening of the performances.
Another example can be found in \cite{DBLP:journals/fcsc/FanLMTW12}, where extensions of RLPM and their application in network analysis and graph databases are discussed.
A mixed problem between SGI and RLPM is presented in \cite{DBLP:journals/jacm/BarceloLR14}.

This situation would benefit from a more general notion of graph pattern matching, able to subsume naturally the more specific ones find in literature.
This general notion would be a common ground to study specific problems and their relationships, as well as to develop common techniques for them.
Moreover, a more general pattern matching notion would pave the way for more general algorithms, which would deal more efficiently with ``mixed'' problems.

To this end, in this paper we propose a new notion about labelled graphs, called \emph{loose graph simulation} (LGS, \cref{sec:LOOSE_SIMULATION}).
The semantics of its pattern queries allow us to check properties from different classical notions of pattern matching, at once and without cumbersome encodings.
LGS queries have a natural graphical representation that simplifies the understanding of their semantic; moreover, they can be composed using a sound and complete algebra (\cref{sec:GUEST_NOTATION_ALGEBRA}).
Various notions of graph pattern matching can be naturally reduced to LGSs, as we will formally prove in \cref{sec:LS_EMPTINESS_PROBLEM,sec:GraphSim_via_LS,sec:RE_via_LS}; in particular, the encoding of subgraph isomorphism allows us to prove that computing LGSs is an NP-complete problem.
Moreover, ``mixed'' matching problems can be easily represented as LGS queries; in fact, these problems can be obtained compositionally from simpler ones by means of the query algebra, as we will show in \cref{sec:REGraphs} where we solve a simplified version of the problem in \cite{DBLP:journals/jacm/BarceloLR14}. Lastly (\cref{sec:POLY}), we study a polynomial-time fragment of LGS that can still be used to compute various notions of graph pattern matching.
Final conclusions and directions for further work (such as a distributed algorithm for computing LGSs) are in \cref{sec:CONCLUSIONS}.

\section{Hosts, guests and loose graph simulations}
\label{sec:LOOSE_SIMULATION}
Loose graph simulations are a generalization of pattern matching for certain labelled graphs.
As often proposed in the literature, the structures that need to be checked for properties are called \emph{hosts}, whereas the structures that represent said properties are called \emph{guests}.

\begin{definition}
	\label{def:multigraph}
	A \emph{host graph} (herein also simply called \emph{graph}) is a triple $(\Sigma,V,E)$ consisting of a finite set of symbols $\Sigma$ (also called \emph{alphabet}), a finite set $V$ of nodes and a set $E \subseteq V \times \Sigma \times V$ of edges.
	For an edge $e = (v,l,v')$ write $s(e)$, $\sigma(e)$, and $t(e)$ for its \emph{source node} $v$, \emph{label} $l$, and \emph{target node} $v'$, respectively. For a vertex $v$ write $\mathrm{in}(v)$ and $\out(v)$ for the sets $\{ e \mid t(e) = v \}$ and $\{ e \mid s(e) = v \}$ of its incoming and outgoing edges.
\end{definition}
\begin{figure}[t!]
	\centering
	\begin{subfigure}[!h]{0.65\textwidth}
	\begin{framed}
		\begin{subfigure}[t]{0.3\textwidth}
			\centering
			\begin{tikzpicture}[transform shape]
			\node [ssplace] (p1) [] {$v$:\MustSym};
			\path [name path=arc] (p1) circle (0.8cm);
			\end{tikzpicture}
			\caption{$v \in \Must$}
		\end{subfigure}%
		~
		\begin{subfigure}[t]{0.3\textwidth}
			\centering
			\begin{tikzpicture}[transform shape]
			\node [ssplace] (p1) [] {$v$:\UniqueSym};
			\path [name path=arc] (p1) circle (0.8cm);
			\end{tikzpicture}
			\caption{$v \in \Unique$}
		\end{subfigure}
		~
		\begin{subfigure}[t]{0.3\textwidth}
			\centering
			\begin{tikzpicture}[transform shape]
			\node [ssplace] (p1) [] {$v$:\InjectiveSym};
			\path [name path=arc] (p1) circle (0.8cm);
			\end{tikzpicture}
			\caption{$v \in \Injective$}
		\end{subfigure}

		\bigskip

		\begin{subfigure}[b]{0.45\textwidth}
			\centering
			\begin{tikzpicture}[baseline=-0.5ex, transform shape]
			\node [ssplace] (p1) [] {$v$};

			\node (e) [above left of=p1] {};
			\path [name path=arc] (p1) circle (0.8cm);
			\path [name path=pa](p1) -- (e);

			\path [name intersections={of=pa and arc, by=E}];
			\node (EL) [left=0.2cm of E] {};
			\node (ER) [right=0.2cm of E] {};

			\path [blackline] (EL) edge [] (ER);
			\path [blackline] (p1) edge [bend left=30] (E);
			\end{tikzpicture}
			\caption{$\emptyset \in \Choice(v)$}
		\end{subfigure}
		~
		\begin{subfigure}[b]{0.45\textwidth}
			\centering
			\begin{tikzpicture}[baseline=2ex,node distance=2cm, transform shape]
			\node [ssplace] (p1) [] {$v$};
			\node [transplace] (f1) [above right of=p1] {};
			\node [transplace] (f2) [below right=4pt of f1] {};
			\node [transplace] (f3) [right of=p1] {};

			\path (f2) -- node[auto=false]{\ldots} (f3);

			\draw [blackarrow, name path=l1] (p1) -- node[above left=-2pt] {$e_1$} (f1);

			\draw [blackarrow, name path=l2] (p1) -- node[above=-5pt] {$e_2$} (f2);

			\draw [blackarrow, name path=l3] (p1) -- node[below right=-2pt] {$e_n$} (f3);

			\path [name path=arc] (p1) circle (0.8cm);

			\path [name intersections={of=l1 and arc,by=E}];
			\draw [blackdot] (E) circle (2pt);

			\path [name intersections={of=l2 and arc,by=F}];
			\draw [blackdot] (F) circle (2pt);

			\path [name intersections={of=l3 and arc,by=G}];
			\draw [blackdot] (G) circle (2pt);

			\draw[blackline] ([shift=(G)]0,0) arc (0:45:0.8cm);
			\end{tikzpicture}
			\caption{${\{e_1,\dots, e_n\} \in \Choice(v)}$}
		\end{subfigure}
	\end{framed}
	\end{subfigure}
	\begin{subfigure}[!h]{0.30\textwidth}
		\centering
		\begin{tikzpicture}[node distance=2cm]


		\node [ssplace] (a) [] {$u$:\MustSym};
		\node [ssplace] (b) [above of=a] {$v$:\MustSym};
		\node (void) [left of=a] {};

		\path [name path=aC] (a) circle (0.55cm);
		\path [name path=bC] (b) circle (0.7cm);

		\draw [blackarrow]
		(a) edge [in=250, out=180, loop] node[below left=-3pt] {$a$} (a);
		\draw [blackarrow, name path=ab]
		(a) -- node[right=-3pt] {$b$} (b);
		\path [name path=aa]
		(a) -- (void);

		\path [name intersections={of=aa and aC,by=Aa}];
		\draw [blackdot] (Aa) circle (2pt);

		\path [name intersections={of=ab and aC,by=Ab}];
		\draw [blackdot] (Ab) circle (2pt);

		\draw[blackline] ([shift=(Ab)]0,0) arc (90:180:0.55cm);

		\node (e) [above right of=b] {};
		\path [name path=pa] (b) -- (e);

		\path [name intersections={of=pa and bC, by=E}];
		\node (EL) [left=0.2cm of E] {};
		\node (ER) [right=0.2cm of E] {};

		\path [blackline] (EL) edge [] (ER);
		\path [blackline] (b) edge [bend right=30] (E);
		\end{tikzpicture}
	\end{subfigure}

	\caption{The guest graphic notation (left) and an example (right).}\label{fig:guest_graphic_notation}
\end{figure}
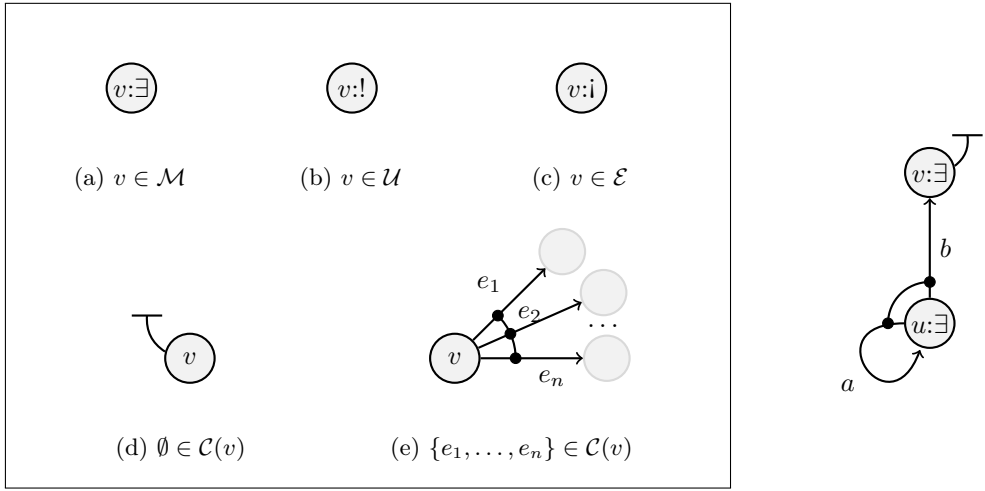
\begin{definition}
	\label{def:guest}
	A \emph{guest} $G = (\Sigma, V,E,\Must,\Unique,\Injective,\Choice)$ is a (host) graph $(\Sigma,V,E)$ additionally equipped with:
	\begin{itemize}
		\item three sets $\Must, \Unique,\Injective \subseteq V$, called respectively \emph{must}, \emph{unique} and \emph{exclusive} set.
		\item a \emph{choice} function $\Choice: V \to \Pwset(\Pwset(E))$, s.t.~$\bigcup \Choice(v) = \out(v)$ for each $v \in V$.
	\end{itemize}
\end{definition}
Roughly speaking, a guest is graph whose:
\begin{itemize}
\item nodes are decorated with usage constraints telling whether they must appear in the host, if their occurrence should be unique, and whether their occurrences can also be occurrences of other nodes or are exclusive;
\item edges are grouped into possible ``choices of sets of ongoing edges'' for any given source node to be considered by a simulation.
\end{itemize}
The semantics of the three sets $\Must$, $\Unique$, $\Injective$ and the choice function $\Choice$ will be presented formally in the definition of loose graph simulations (\cref{def:loose_simulation}).

Guests can be conveniently represented using the graphical notation shown in \cref{fig:guest_graphic_notation} (a formal algebra is discussed in \cref{sec:GUEST_NOTATION_ALGEBRA}).
A node belonging to the must, unique or exclusive set is decorated with the symbols $\exists$, $\UniqueSym$ and $\InjectiveSym$, respectively.
Choice sets are represented by arcs with dots placed on the intersection with each edge that belongs to the given choice set. The empty empty choice set ($\emptyset \in \Choice(v))$ is represented by the ``corked edge''
$(\!\!\begin{tikzpicture}[thick,scale=0.7, every node/.style={transform shape}]
\node (p1) [minimum size=0mm, inner sep=0pt] {};

\node (E) [above left=0.4cm of p1, minimum size=0mm, inner sep=0pt] {};

\node (EL) [left=0.2cm of E] {};
\node (ER) [right=0.2cm of E] {};

\path [blackline] (EL) edge [] (ER);
\path [blackline] (p1) edge [bend left=60] (E);
\end{tikzpicture})$.

\begin{example}
\cref{fig:guest_graphic_notation} shows the graphical representation of a guest with two nodes $u$ and $v$. The must set is $\{u,v\}$, the unique and exclusive sets are both empty, and the choice function takes $u$ to $\{\{(u,a,u),(u,b,v)\}\}$ and $v$ to $\{\emptyset\}$.
\end{example}
Before we formalise the notion of loose graph simulation, we need some auxiliary definitions. The following one fix the notation for paths in a graph.

\begin{definition}\label{def:pathdomain}
	For $M = (\Sigma, V,E)$, define $\Path_M$ as the set $\bigcup_{n\in\N}\{(e_0,\dots,e_n) \in E^n \mid \forall i \in \{1,\dots,n\}\ s(e_i) = t(e_{i-1})\}$ of all paths in $M$.
	Source (${s\colon \Path_M \to V}$), target ($t\colon \Path_M \to V$), and label ($\sigma\colon\Path_M \to \Sigma^+$) functions are extended accordingly:
	$s((e_0,\dots,e_n)) \triangleq s(e_0)$,
	$t((e_0,\dots,e_n)) \triangleq t(e_n)$, and
	$\sigma((e_0,\dots,e_n)) \triangleq \sigma(e_0)\dots\sigma(e_n)$.
	Lastly, for any $v,v^\prime \in V$,
	define $\Path_M(v,v^\prime)$ as the set of all paths from $v$ to $v^\prime$, formally  $\Path_M(v,v^\prime)\triangleq\{ \rho \in \Path_M \mid s(\rho) = v \land t(\rho) = v^\prime \}$.
\end{definition}

Akin to graph simulations (\cref{def:graph_simulation}), LGSs are subgraphs of the product of guest and host that are coherent with the additional data prescribing node and edge usage.

\begin{definition}\label{def:graph_product}
	Let $M_1 = (\Sigma_1, V_1, E_1)$ and $M_2 = (\Sigma_2, V_2, E_2)$ be two graphs. The \emph{tensor product graph} $M_1 \times M_2$ is the graph $(\Sigma_1 \cap \Sigma_2,\ V_1 \times V_2,\ E^\times)$ where $E^\times\triangleq \{((u,u^\prime),a,(v,v^\prime)) \mid (u,a,v) \in E_1 \land (u^\prime, a, v^\prime) \in E_2 \}$.
\end{definition}
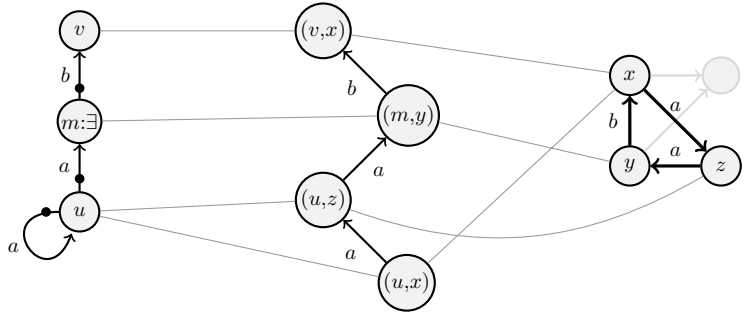
\begin{figure}[t]
	\centering
	\begin{tikzpicture}[scale=0.8, node distance=1.5cm, transform shape]

	\node [ssplace, align=center] (m) [] {$m$:\MustSym};
	\node [ssplace] (f1) [below of=m] {$u$};
	\node [ssplace] (f2) [above of=m] {$v$};
	\node (hd) [left of=f1] {};

	\draw [blackarrow] (f1) edge [in=250, out=180, loop] node (uau) {} (f1);
	\node [left=-8pt of uau] {$a$};
	\path [name path=f1f1] (f1) -- (hd);

	\path [blackarrow, name path=mf2] (m) -- node (mbv) {} (f2);
	\node [left=-8pt of mbv] {$b$};

	\path [blackarrow, name path=f1m] (f1) -- node (uam) {} (m);
	\node [left=-8pt of uam] {$a$};

	\path [name path=f1C] (f1) circle (0.55cm);
	\path [name path=mC] (m) circle (0.55cm);

	\path [name intersections={of=mC and mf2,by=mD}];
	\draw [blackdot] (mD) circle (2pt);

	\path [name intersections={of=f1C and f1f1,by=f1D}];
	\draw [blackdot] (f1D) circle (2pt);

	\path [name intersections={of=f1C and f1m,by=f1D2}];
	\draw [blackdot] (f1D2) circle (2pt);

	\node (e) [above left of=f2] {};
	\path [name path=arc] (f2) circle (0.8cm);
	\path [name path=pa](f2) -- (e);

	\node [ssplace] (s1) [right=3.2cm of f2] {($v$,$x$)};
	\node [ssplace] (s2) [below right=1cm of s1] {($m$,$y$)};
	\node [ssplace] (s3) [below left=1cm of s2] {($u$,$z$)};
	\node [ssplace] (s4) [below right=1cm of s3] {($u$,$x$)};

	\path [blackarrow] (s4) -- node (nn1) {} (s3);
	\path [blackarrow] (s3) -- node (nn2) {} (s2);
	\path [blackarrow] (s2) -- node (nn3) {} (s1);

	\node [below left=-12pt of nn1] {$a$};
	\node [below right=-12pt of nn2] {$a$};
	\node [below left=-12pt of nn3] {$b$};


	\node (hh2) [right=3cm of s2]	{};
	\node [ssplace] (n2) [above=0.2cm of hh2]	{$x$};
	\node [transplace] (h1)  [right of=n2] {};
	\node [ssplace] (m2) [below of= n2]	{$y$};

	\node [ssplace] (n1b) [right of= m2]   {$z$};

	\path[transplace]
	(n2) edge [->] (h1)
	(m2) edge [->] (h1);

	\path[blackarrow, black, very thick]
	(n1b) edge [->] node (yaw) {} (m2)
	(m2) edge [->] node (wbz) {} (n2)
	(n2) edge [->] node (nab) {} (n1b);

	\node [left=-7pt of wbz] {$b$};
	\node [above=-8pt of yaw] {$a$};
	\node [above=-8pt of nab] {$a$};


	\draw[black!40]
	(n2) -- (s4)
	(m2) -- (s2)
	(n1b) edge[bend left=25] (s3)
	(n2) -- (s1)
	;

	\draw[black!40]
	(f1) -- (s4)
	(f1) -- (s3)
	(m) -- (s2)
	(f2) -- (s1)
	;

	\end{tikzpicture}
	\label{fig:sfig1}

	\caption{An LGS (center) between a guest (left) and a host (right).}
	\label{fig:loose_simulation}
\end{figure}

When clear from the context, we denote host graphs and their components as $H$ and as ${(\Sigma_H, V_H,E_H)}$ (and variations thereof).
We adopt the convention of denoting guests as $G$ (and variations thereof) and writing ${(\Sigma_G, V_G,E_G,\Must,\Unique, \Injective, \Choice)}$ for the components of the guest $G$. We are now ready to define the notion of loose graph simulation.

\begin{definition}
	\label{def:loose_simulation}
	A \emph{loose graph simulation} (LGS for short) of $G$ in $H$ is a subgraph $(\Sigma_G \cap \Sigma_H,V^{G\to H},E^{G \to H})$ of $G \times H$ subject to the following conditions:
	\begin{conditions}{LGS}{ls:c}
		\item vertices of $G$ in the must set occur in $V^{G\to H}$, \ie
		for each $u \in \Must$ there exists $u^\prime \in V_H$ such that $(u,u^\prime) \in V^{G\to H}$;
		\item vertices in the unique set are assigned to at most one vertex of $H$, \ie
		for each $u \in \Unique$ and all $u^\prime,v^\prime \in V_H$, if $(u,u^\prime) \in V^{G\to H}$ and $(u,v^\prime) \in V^{G\to H}$ then $u^\prime = v^\prime$;
		\item vertices of $H$ assigned to a vertex in the exclusive set cannot be assigned to other vertices, \ie
		for each $u \in \Injective$, $v \in V_G$ and $u^\prime \in V_H$, if $(u,u^\prime) \in V^{G\to H}$ and $(v,u^\prime) \in V^{G\to H}$ then $u = v$;
		\item for $(u,u^\prime) \in V^{G\to H}$, there is a set in $\Choice(u)$ s.t.~each of its elements is related to an edge with source $u^\prime$ and only such edges occur in $E^{G\to H}$. Formally,
		\begin{itemize}
			\let\item\olditem
			\item for each 	$(u,u^\prime) \in V^{G\to H}$ there exists $\gamma \in \Choice(u)$ such that for all $(u,a,v) \in \gamma$ it holds that $((u,u^\prime),a,(v,v^\prime)) \in E^{G\to H}$ for some $v^\prime \in V_H$;
			\item for each $((u,u^\prime),a,(v,v^\prime)) \in E^{G\to H}$ there exists $\gamma \in \Choice(u)$ s.t.~${(u,a,v) \in \gamma}$ and for each $(u,b,w) \in \gamma$ it holds that $((u,u^\prime),b,(w,w^\prime)) \in E^{G\to H}$ for some $w^{\prime} \in V_H$.
		\end{itemize}
		\item the simulation preserves the connectivity with respect to nodes marked as must:
		for each $(u,u^\prime) \in V^{G\to H}$ and $v \in \Must$ if $\Path_G(u,v) \neq \emptyset$ then there exists $v^\prime \in V_H$ such that $\Path_{(\Sigma_G \cap \Sigma_H,V^{G\to H},E^{G \to H})}((u,u^\prime),(v,v^\prime)) \neq \emptyset$.
	\end{conditions}
	The domain of all LGSs for $G$ and $H$ is denoted as $\SSD{G}{H}$.
\end{definition}

\looseness=-1
As already mentioned at the end of \cref{def:guest}, the definition of LGS attributes a semantics for the must, unique, exclusive sets and the choice function.
Regarding the \emph{unique set}, \cref{ls:c2} requires that every vertex of the guest in this set to be mapped by at most one element of the host. Similarly, \cref{ls:c3} requires the vertices of the host paired in the LGS with a node of the \emph{exclusive set} to be only paired with that node.
\Cref{ls:c4} defines the semantics of the choice function: given a pair of vertices $(u,u^\prime) \in V^{G \to H}$, it requires to select at least one set from $\Choice(u)$. The edges of these selected sets (and only these edges, as stated by the second part of the condition) must be paired in the LGS to edges in $H$ with source $u^\prime$.
This condition can be seen as a generalization of the second condition of \emph{graph simulations} (\cref{def:graph_simulation}) that requires all outgoing edges from $u$ to be in relation with outgoing edges of $u^\prime$.

\Cref{ls:c1,ls:c5} formalise the constraints attached to must nodes: the first condition imposes that every vertex in this set must appear in the LGS, while the second condition requires that, for each $(u,u^\prime) \in V^{G \to H}$, each vertex in the must set reachable in the guest from $u$ is also reachable in the LGS, with a path starting from $(u,u^\prime)$.

\begin{example}
	\Cref{fig:loose_simulation} shows a guest and its loose graph simulation over a host. In this example $\Must = \{m\}$ and $\Unique = \Injective = \emptyset$.
	Moreover, the choice function is \emph{linear}, \ie for each vertex $u$, $\Choice(u)$ contains a set $\{e\}$ for each edge in $\out(u)$ and $\emptyset$ whenever $\out(u) = \emptyset$, formally $\Choice = \lambda x. \{\{e\} \mid e \in \out(x)\} \cup \{\emptyset \mid \out(x) = \emptyset \}$.
	LGSs of this guest represents paths $(e_0,e_1,\dots,e_n)$ of arbitrary length in the host such that $\forall i < n\ \sigma(e_i) = a$ and $\sigma(e_n) = b$.
	The guest is therefore similar to the regular language $a^\star b$ and a LGS identifies paths in the host labelled with words in this language.
\end{example}

\begin{proposition}\label{prop:ls5_elim}
	Let $G$ be a guest with choice function $\Choice$ defined as $\lambda x.\{\out(x)\}$, let $H$ be a host and let $S = (\Sigma_G \cap \Sigma_H,V^{G\to H},E^{G \to H})$ be a subgraph of $G \times H$. If $S$ satisfies \cref{ls:c4} then it also satisfies \cref{ls:c5}.
\end{proposition}

\begin{proof}
	Let $\Choice(v) = \{\out(v)\}$ for all $v \in V_G$. If $(u,u^\prime) \in V^{G\to H}$ then \ref{ls:c4} requires that for all $(u,a,v) \in \out(u)$ there exists $v^\prime$ such that $(v,v^\prime) \in V^{G \to H}$ and $((u,u^\prime),a,(v,v^\prime)) \in E^{G \to H}$.
	Coinductively, since the same will hold for every of those pair $(v,v^\prime)$, it follows that whenever there is a path in $G$ from $u$ to a node $m \in \Must$ in the must set, then there must be a path in $S$ from $(u,u^\prime)$ to a pair of vertices $(m,w)$, where $w \in V_H$. Hence, \ref{ls:c5} holds.
\end{proof}

\section{An algebra for guests}
\label{sec:GUEST_NOTATION_ALGEBRA}
Guests are used to specify the patterns to look for inside a host; hence they should be easy to construct and to understand. To this end, besides the graphical notation described in \cref{sec:LOOSE_SIMULATION}, in this section we introduce an algebra for guests which allows us to construct them in a compositional way.

\begin{definition} 
	A guest is \emph{empty} whenever it has no vertexes.
	A guest with only one vertex and no edges is a \emph{unary guest} and is denoted as
	\[
		p_\Attr \triangleq (\emptyset, \{p\}, \emptyset, \{p \mid \MustSym \in \Attr\}, \{p \mid \UniqueSym \in \Attr\}, \{p \mid \InjectiveSym \in \Attr\}, \{p \to \{\emptyset \mid \varnothing \in \Attr\}\})
	\]
	where $p$ is the only vertex and $\Attr \subseteq \{\MustSym,\UniqueSym,\InjectiveSym,\varnothing\}$ state if $p$ is respectively in $\Must$, $\Unique$, $\Injective$ or if $\emptyset \in \Choice(p)$.
	For $\alpha$ a name, $P$ and $Q$ unary guests,
	the \emph{arrow operator} from $P$ to $Q$  $\alpha$ is defined as
	\begin{gather*}
	P \xrightarrow{\alpha} Q \triangleq (\{\alpha\}, \{p,q\}, \{(p,\alpha,q)\}, \Must_P \cup \Must_Q, \Unique_P \cup \Unique_Q, \Injective_P \cup \Injective_Q, \Choice^\rightarrow)\\
	\Choice^\rightarrow \triangleq \lambda x.\begin{cases}
	c_P \cup \{\{(p,\alpha,q)\}\} \cup c_Q &\text{if } p = q \land x = p\\
	c_P \cup \{\{(p,\alpha,q)\}\} &\text{if } p \neq q \land x = p\\
	c_Q &\text{if } p \neq q \land x = q
	\end{cases}
	\end{gather*}
	A guest is called \emph{elementary} whenever it is empty, unary, or the result of the arrow operator.
\end{definition}

For example, a node $p$ with only a self loop labelled $\alpha$ can be expressed with the term $p \xrightarrow{\alpha} p$.
Besides the elementary guests, the algebra is completed by introducing two binary operators used to combine guests.

\begin{definition} 
	Let $G_1$
	and $G_2$
	be two guests.
	Their \emph{addition} is the guest:
	\begin{equation*}
	G_1 \oplus G_2 \triangleq (\Sigma_1 \cup \Sigma_2, V_1 \cup V_2, E_1 \cup E_2, \Must_1 \cup \Must_2, \Unique_1 \cup \Unique_2, \Injective_1 \cup \Injective_2, \Choice^{\oplus})
	\end{equation*}
	where the choice function $\Choice^\oplus$ is defined as
	\begin{equation*}
	\Choice^\oplus \triangleq \lambda x.\begin{cases}
	\Choice_1(x)\cup\Choice_2(x) &\text{if } x \in V_1 \land x \in V_2\\
	\Choice_1(x) &\text{if } x \in V_1\\
	\Choice_2(x) &\text{if } x \in V_2
	\end{cases}
	\end{equation*}
	The \emph{multiplication} of $G_1$ and $G_2$ is the guest:
	\begin{equation*}
		G_1 \otimes G_2 \triangleq (\Sigma_1 \cup \Sigma_2, V_1 \cup V_2, E_1 \cup E_2, \Must_1 \cup \Must_2, \Unique_1 \cup \Unique_2, \Injective_1 \cup \Injective_2, \Choice^{\otimes})
	\end{equation*}
	where the choice function $\Choice^\otimes$ is defined as follows
	\begin{equation*}
	\Choice^\otimes \triangleq \lambda x.\begin{cases}
	\{\gamma_1\cup \gamma_2 \mid \gamma_1 \in \Choice_1(x) \land \gamma_2 \in \Choice_2(x)\} &\text{if } x \in V_1 \land x \in V_2\\
	\Choice_1(x) &\text{if } x \in V_1\\
	\Choice_2(x) &\text{if } x \in V_2
	\end{cases}
	\end{equation*}
\end{definition}
Notice how addition and multiplication operators differ only by the definition of the choice function for vertices of both $G_1$ and $G_2$.
In the case of addition, the resulting choice function is the union of the two choice function $\Choice_1$ and $\Choice_2$, whereas for the multiplication, given a vertex $v \in V_1 \cap V_2$, every set of $\Choice^\otimes(v)$ is the union of a set in $\Choice_1(v)$ and one in $\Choice_2(v)$.
\begin{proposition}
	The operations $\oplus$ and $\otimes$ form an idempotent commutative semiring structure over the set of all guests.
\end{proposition}
The algebra offers a clean and modular representation of guests. Modularity, in particular, allows us to combine queries as illustrated in the second part of this work. Furthermore, guests admit normal forms.

\begin{definition} 
	A term $G$ in the algebra of guests is in \emph{normal form} if $G = \bigoplus_{i \in I}\bigotimes_{j \in J_i} G_{i,j}$ where each $G_{i,j}$ is an elementary guest.
\end{definition}

\begin{example}
	Consider the guest \[
		(\{a,b\},\{p,q\},\{(p,a,p),(p,b,q)\},\{p,q\},\emptyset,\emptyset,\{p \mapsto \{\{(p,a,p),(p,b,q)\}\}, q \mapsto \{\emptyset\} \})
	\]
	shown in \cref{fig:guest_graphic_notation} on the right. This guest is represented by the term $q_{\{\exists,\varnothing\}} \oplus ( p_{\{\exists\}} \xrightarrow{a} p \otimes p \xrightarrow{b} q )$ which is in normal form.
\end{example}

\begin{proposition}
	For $G = (\Sigma,V,E,\Must,\Unique,\Injective,\Choice)$ a guest, its normal form is:
	\begin{equation*}
		\bigoplus_{v \in V} v_{\{\exists \mid v \in \Must\} \cup \{\UniqueSym \mid v \in \Unique \} \cup \{ \InjectiveSym \mid v \in \Injective \} \cup \{ \varnothing \mid \emptyset \in \Choice(v) \}} \oplus  \bigoplus_{\substack{v \in V\\ \gamma \in \Choice(v)}}\left(\bigotimes_{e \in \gamma} \left( s(e) \xrightarrow{\sigma(e)} t(e)\right)\right)
	\end{equation*}
\end{proposition}

For $G = (\Sigma, V, E, \Must, \Unique, \Injective, \Choice)$ a guest, we write
$G[p/q]$ for the guest obtained renaming $p \in V$ as $q \not\in V$. In particular, the set of edges and choice function are:
\begin{align*}
E[p/q] &= \left\{(u,a,v)\,\middle|\,\begin{array}{l}
	(u^\prime,a,v^\prime) \in E\text{, }
	(u^\prime \neq p \implies u = u^\prime)\text{, } (u^\prime = p \implies u = q)\text{,}\\
	(v^\prime \neq p \implies v = v^\prime) \text{, and } (v^\prime = p \implies v = q)
	\end{array}\right\}\\
\Choice[p/q] &= \lambda x.\begin{cases}
\{S[p/q] \mid S \in \Choice(x)\} &\text{if } x \neq p \land x \neq q\\
\{S[p/q] \mid S \in \Choice(p)\} &\text{if } x = q\\
\end{cases}
\end{align*}

\section{The LGS problem is NP-complete}\label{sec:LS_EMPTINESS_PROBLEM}

In this section we analyse the complexity of computing LGSs by studying their emptiness problem.
Without loss of generality, we restrict to guests and hosts with the same $\Sigma$. In the following, let $G = (\Sigma_G, V_G,E_G,\Must,\Unique, \Injective, \Choice)$ and $H = (\Sigma_H, V_H,E_H)$ be a guest and a host respectively.

\begin{definition} 
	\label{def:ls_emptiness_problem}
	The \emph{emptiness problem for LGSs} for $G$ and $H$ consists in checking $\SSD{G}{H} = \emptyset$.
\end{definition}

\begin{proposition}\label{theo:LS_check_in_P}
	Computing LGSs, as well as their emptiness problem, is in NP.
\end{proposition}
\begin{proof}
	Let ${S = (\Sigma, V^{G\to H}, E^{G \to H})}$ be a subgraph of $G\times H$. We will now prove that there exists a polynomial algorithm w.r.t.~the size of $G$ and $H$ that checks whether $S$ satisfies all the conditions of \cref{def:loose_simulation}.
	The satisfiability checking of \cref{ls:c1} is in $\BO(\Must \times V^{G\to H})$ since it is sufficient for every vertex in the must set $\Must$ to check whether there is a vertex of the host paired with it.
	For similar reasons, \cref{ls:c2,ls:c3} can also be checked in polynomial time.
	Moreover, to check \cref{ls:c4} it is sufficient to check, for each $(u,v) \in V^{G \to H}$, whether there is $\gamma \in \Choice(v)$ s.t.~$\gamma \subseteq \pi_1 \circ \out((u,v))$ and if for all $u^\prime \in \pi_1 \circ \out((u,v))$ there exists
	$\gamma \in \Choice(v)$ s.t.~$u^\prime \in \gamma \subseteq \pi_1 \circ \out((u,v))$.
	This can be done by a naive algorithm in $\BO(V_H \times E_G \times (V_G \times E_H + \Choice \times E_G^2))$.
	Lastly, checking whether $S$ satisfies \cref{ls:c5} requires the evaluation of the reachability relation of $G$ and $S$ and therefore can be computed in $\BO(V_G^3 \times V_H^3)$ using the Floyd-Warshall Algorithm \cite{DBLP:journals/cacm/Floyd62a}.
	Since every condition can be checked in polynomial time we can conclude that the LGS problem is in NP.
\end{proof}

\subsection{NP-hardness: subgraph isomorphisms via LGSs}
\label{sec:SubIso_via_LS}

We will now show the NP-hardness of the emptiness problem for LGSs by reducing the emptiness problem for subgraph isomorphism to it.
The subgraph isomorphism problem requires to check whether a subgraph of a graph (host) and isomorphic to a second graph (query) exists.
Application of this problem can be found in network analysis \cite{DBLP:conf/sigcomm/LischkaK09}, bioinformatics and chemoinformatics \cite{DBLP:journals/bmcbi/BonniciGPSF13,bioinflmu-308}.

\begin{definition} 
	\label{def:subgraph_isomorphism}
	Let $H = (\Sigma,V_H,E_H)$ and $Q = (\Sigma,V_Q,E_Q)$ be two graphs called \emph{host} and \emph{query} respectively.
	There exists a subgraph of $H$ isomorphic to $Q$ whenever there exists a pair of injections $\phi : V_Q \injarrow V_H$ and $\eta: E_Q \injarrow E_H$ s.t.~	$\sigma(e) = \sigma \circ \eta(e)$,
	$\phi \circ s(e) = s \circ \eta(e)$, and
	$\phi \circ t(e) = t \circ \eta(e)$ for each $e \in E_Q$.
\end{definition}
The subgraph isomorphism problem, as well as the emptiness problem associated to it, is shown to be NP-complete by Cook \cite{DBLP:conf/stoc/Cook71}.
Its complexity and its importance makes it one of the most studied problem and multiple algorithmic solutions where derived for it \cite{DBLP:journals/bmcbi/BonniciGPSF13,DBLP:journals/jacm/Ullmann76,DBLP:journals/pami/CordellaFSV04}.
We will now show that the emptiness problem for subgraph isomorphism can be solved using LGSs.

\begin{proposition}\label{prop:SubIso_via_LS}
	Let $H = (\Sigma,V_H,E_H)$ and $Q = (\Sigma,V_Q,E_Q)$ be a host and a query for subgraph isomorphism respectively. Moreover, let
	\begin{equation*}
	G = \bigoplus_{v \in V_Q} v_{\{\exists\UniqueSym\InjectiveSym\}\cup\{\varnothing \mid \out(v) = \emptyset \}} \oplus \left(\bigotimes_{e \in E_Q} \left(s(e) \xrightarrow{\sigma(e)} t(e)\right)\right)
	\end{equation*}
	Then, there exists a subgraph of $H$ isomorphic to $Q$ iff there is a LGS of $G$ in $H$, \ie $\SSD{G}{H} \neq \emptyset$.
\end{proposition}
\begin{proof}
		From the definition of $G$, its must, unique and exclusive sets, as well as its choice function, are $\Must = \Unique = \Injective = V_Q$ and  $\Choice = \lambda x.\{\out(x)\}$ respectively.
	 Suppose $\phi : V_Q \injarrow V_H$ and $\eta: E_Q \injarrow E_H$ be two injections as in \cref{def:subgraph_isomorphism}.
	 Then the graph $S = (\Sigma, V^{G\to H}, E^{G \to H})$ where $V^{G\to H} \triangleq \{ (u,u^\prime) \mid u^\prime = \phi(u)\}$ and $E^{G \to H} \triangleq \{ ((u,u^\prime),a,(v,v^\prime)) \mid (u^\prime,a,v^\prime) = \eta((u,a,v))\}$
	 form a LGS for $G$. Indeed, it satisfy \cref{ls:c1,ls:c2,ls:c3}, since $\phi$ is an injection.
	 Moreover, since $\eta: E_Q \injarrow E_H$ is also an injection and for each edge $e \in E_Q$ it holds that $\sigma(e) = \sigma \circ \eta(e)$, $\phi \circ s(e) = s \circ \eta(e)$ and
	 $\phi \circ t(e) = t \circ \eta(e)$,
	 $S$ must be such that for each $(u,u^\prime) \in V^{G\to H}$ and for each $(u,a,v) \in \out(u)$ there exists $v^\prime$ such that $(v,v^\prime) \in V^{G \to H}$ and $((u,u^\prime),a,(v,v^\prime)) \in E^{G \to H}$. It follows that
	 $S$ is a subgraph of $G \times H$ and \cref{ls:c4} is satisfied, since $\Choice(u) = \{\out(u)\}$. Moreover the satisfaction of \cref{ls:c5} follows from \cref{prop:ls5_elim}. $S$ is therefore a LGS of $G$ in $H$.
	 Conversely, suppose that there is a LGS $S = (\Sigma, V^{G\to H}, E^{G \to H})$. Let $\phi$ s.t.~$\phi(u) = u^\prime \iff (u,u^\prime) \in V^{G \to H}$ and $\eta$ s.t.~$
		 \eta((u,a,v)) = (u^\prime,a,v^\prime) \iff ((u,u^\prime),a,(v,v^\prime)) \in E^{G \to H}$.
	 Since $\Must = \Unique = \Injective = V_Q$ and $S$ is a LGS, it holds that $\phi$ is an injection defined on the domain $V_Q$. Moreover $\eta$ is also an injection, since $\Choice = \lambda x. \{\out(x)\}$ and $S$ satisfies \cref{ls:c4}, and together with the hypothesis that $S$ is a subgraph of $G\times H$ it must also hold that for each edge
	 $e \in E_Q$ $\sigma(e) = \sigma \circ \eta(e)$,
	 $\phi \circ s(e) = s \circ \eta(e)$ and
	 $\phi \circ t(e) = t \circ \eta(e)$. There is therefore a subgraph of $H$ isomorphic to $Q$.
\end{proof}

\begin{figure}[t!]
	\centering
		\begin{tikzpicture}[node distance=1.5cm]


		\node [ssplace] (a) [] {};
		\node [ssplace] (b) [above right of=a] {};
		\node [ssplace] (c) [above left of=a] {};
		\node [ssplace] (d) [below left of=a] {};

		\path [blackarrow]
		(a) edge [->] node[below right=-3pt] {$b$} (b)
		(c) edge [->] node[above=-3pt] {$b$} (b)
		(c) edge [->] node[above right=-3pt] {$a$} (a)
		(a) edge [->] node[below right=-3pt] {$a$} (d)
		(d) edge [->] node[left=-3pt] {$a$} (c);


		\node [ssplace] (a1) [right=5.5cm of a] {\MustSym\UniqueSym\InjectiveSym};
		\node [ssplace] (b1) [above right of=a1] {\MustSym\UniqueSym\InjectiveSym};
		\node [ssplace] (c1) [above left of=a1] {\MustSym\UniqueSym\InjectiveSym};
		\node [ssplace] (d1) [below left of=a1] {\MustSym\UniqueSym\InjectiveSym};

		\path [name path=aC] (a1) circle (0.7cm);
		\path [name path=bC] (b1) circle (0.7cm);
		\path [name path=cC] (c1) circle (0.7cm);
		\path [name path=dC] (d1) circle (0.7cm);

		\draw [blackarrow, name path=ab]
		(a1) -- node[below right=-3pt] {$b$} (b1);
		\draw [blackarrow, name path=cb]
		(c1) -- node[above=-3pt] {$b$} (b1);
		\draw [blackarrow, name path=ca]
		(c1) -- node[above right=-3pt] {$a$} (a1);
		\draw [blackarrow, name path=ad]
		(a1) -- node[below right=-3pt] {$a$} (d1);
		\draw [blackarrow, name path=dc]
		(d1) -- node[left=-3pt] {$a$} (c1);

		\path [name intersections={of=aC and ab,by=Ab}];
		\draw [blackdot] (Ab) circle (2pt);

		\path [name intersections={of=ad and aC,by=Ad}];
		\draw [blackdot] (Ad) circle (2pt);

		\draw[blackline] ([shift=(Ab)]0,0) arc (40:-125:0.7cm);

		\path [name intersections={of=cC and cb,by=Cb}];
		\draw [blackdot] (Cb) circle (2pt);

		\path [name intersections={of=cC and ca,by=Ca}];
		\draw [blackdot] (Ca) circle (2pt);

		\draw[blackline] ([shift=(Cb)]0,0) arc (0:-45:0.7cm);

		\path [name intersections={of=dC and dc,by=Dc}];
		\draw [blackdot] (Dc) circle (2pt);

		\node (e) [above right of=b1] {};
		\path [name path=pa] (b1) -- (e);

		\path [name intersections={of=pa and bC, by=E}];
		\node (EL) [left=0.2cm of E] {};
		\node (ER) [right=0.2cm of E] {};

		\path [blackline] (EL) edge [] (ER);
		\path [blackline] (b1) edge [bend right=30] (E);
		\end{tikzpicture}
	\caption{A possible query for subgraph isomorphism (on the left) and its translation to a guest for LGSs (on the right).}\label{fig:subgraph_isomorphism}
\end{figure}
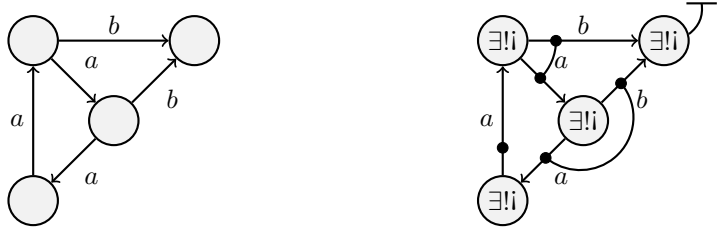

Note how the translation from subgraph isomorphism's queries to guest for LGSs defined in \cref{prop:SubIso_via_LS} is \emph{structure-preserving}.
Indeed, an example of this can be seen in \cref{fig:subgraph_isomorphism}.
This property is important since it makes defining LGSs' guests to solve the subgraph isomorphism problem as intuitive as the respective queries for it. This is also the case for other notions commonly used in the graphs' pattern matching community.
Moreover, since the translated guest is as intuitive as the original query, this property strengthens the idea of using guests and LGSs to represent and compute hybrid queries w.r.t.~these notions.

From \cref{theo:LS_check_in_P} and \cref{prop:SubIso_via_LS} it follows that:
\begin{theorem}
	The emptiness problem for LGSs is NP-complete.
\end{theorem}

\section{Graph simulations are loose graph simulations}\label{sec:GraphSim_via_LS}

Graph simulations are particular relations between graphs that are extensively applied in several fields
\cite{DBLP:conf/icdt/Fan12,DBLP:journals/pvldb/FanWWD14}.
The \emph{graph simulation problem} requires to check whether a portion of a graph (host) \emph{simulates} another graph (query).

\begin{definition}\label{def:graph_simulation}
	A \emph{graph simulation} of $Q = (\Sigma,V_Q,E_Q)$ (herein \emph{query}) in $H = (\Sigma,V_H,E_H)$ (herein \emph{host}) is a relation ${\R \subseteq V_Q \times V_H}$ such that:
	\begin{itemize}
		\item for each node $u \in V_Q$ there exists a node $v \in V_H$ such that $(u,v) \in \R$;
		\item for each pair $(u,v) \in \R$ and for each edge $e \in \out(u)$ there exists an edge $e^\prime \in \out(v)$ such that $\sigma(e) = \sigma(e^\prime)$ and $(t(e),t(e^\prime)) \in \R$.
	\end{itemize}
\end{definition}

Graph simulation existence can be decided in polynomial time \cite{DBLP:journals/scp/BloomP95,DBLP:conf/focs/HenzingerHK95}.
Their emptiness problem can be reduced to the emptiness problem for loose ones.

\begin{proposition}\label{prop:GraphSim_via_LS}
	Let $H = (\Sigma,V_H,E_H)$ and $Q = (\Sigma,V_Q,E_Q)$ be a host and a query for graph simulation respectively. Moreover, let
	\begin{equation*}
		G = \bigoplus_{v \in V_Q} v_{\{\exists\}\cup\{\varnothing \mid \out(v) = \emptyset \}} \oplus \bigotimes_{e \in E_Q} s(e) \xrightarrow{\sigma(e)} t(e)
	\end{equation*}
	Then, there is a graph simulation of $Q$ in $H$ iff $\SSD{G}{H} \neq \emptyset$.
\end{proposition}

\begin{proof}
	From definition of $G$, its must, unique, exclusive sets and its choice function are $\Must = V_Q$, $\Unique = \Injective = \emptyset$ and $\Choice = \lambda x. \{\out(x)\}$ respectively.
	Let $\R$ be a graph simulations. The graph $S = (\Sigma,V^{G \to H}, E^{G \to H})$ where $V^{G \to H} = \R$ and $E^{G \to H} = \{((u,u^\prime),a,(v,v^\prime)) \mid
		(u,u^\prime),(v,v^\prime) \in R, (u,a,v) \in E_Q, (u^\prime,a,v^\prime) \in E_H\}$
	is a loose graph simulations for $G$. $\Unique = \Injective = \emptyset$ makes \cref{ls:c2,ls:c3} always true, whereas the first condition of \cref{def:graph_simulation}, that requires all vertices of $V_Q$ to appear in the first projection of $\R$, makes \cref{ls:c1} satisfied.
	The second condition of \cref{def:graph_simulation} requires that, given a pair $(u,v) \in R$, every edge of $\out(u)$ is associated with one edge of $\out(v)$ with the same label and with targets paired in $\R$. \Cref{ls:c4} is therefore satisfied. Lastly, the satisfaction of \cref{ls:c5} follows from \cref{prop:ls5_elim}. $S$ is therefore a loose graph simulation of $G$ in $H$.
	Conversely, suppose there exists a LGS $S = (\Sigma, V^{G \to H}, E^{G \to H})$. Then $V^{G \to H}$ is a graph simulation. The definition of must set $\Must = V_Q$ ensures that each vertex of $V_Q$ must appear in the first projection of $V^{G \to H}$: the first condition of \cref{def:graph_simulation} is  satisfied.
	Moreover, the definition of the choice function $\Choice = \lambda x.\{\out(x)\}$ and \cref{ls:c4} implies that for each $(u,u^\prime) \in V^{G \to H}$ and for all $(u,a,v) \in \out(u)$ there exists $v^\prime$ such that $((u,u^\prime),a,(v,v^\prime)) \in E^{G \to H}$ and,
	since $S$ is a subgraph of $G \times H$, $(v,v^\prime) \in V^{G \to H}$. Thus, the second condition of \cref{def:graph_simulation} holds and $V^{G \to H}$ is a graph simulation.
\end{proof}
\begin{figure}[t!]
	\centering
		\begin{tikzpicture}[node distance=2cm]


		\node [ssplace] (a0) [] {};
		\node [ssplace] (b0) [right of=a0] {};

		\path [blackarrow]
		(a0) edge [in=110, out=180, loop, ->] node[above left=-3pt] {$a$} (a0)
		(a0) edge [->] node[above=-3pt] {$b$} (b0);


		\node [ssplace] (a) [right=3.5cm of b0] {\MustSym};
		\node [ssplace] (b) [right of=a] {\MustSym};
		\node (void) [left of=a] {};

		\path [name path=aC] (a) circle (0.55cm);
		\path [name path=bC] (b) circle (0.7cm);

		\draw [blackarrow]
		(a) edge [in=110, out=180, loop] node[above left=-3pt] {$a$} (a);
		\draw [blackarrow, name path=ab]
		(a) -- node[above=-3pt] {$b$} (b);
		\path [name path=aa]
		(a) -- (void);

		\path [name intersections={of=aa and aC,by=Aa}];
		\draw [blackdot] (Aa) circle (2pt);

		\path [name intersections={of=ab and aC,by=Ab}];
		\draw [blackdot] (Ab) circle (2pt);

		\draw[blackline] ([shift=(Ab)]0,0) arc (0:-180:0.55cm);

		\node (e) [above right of=b] {};
		\path [name path=pa] (b) -- (e);

		\path [name intersections={of=pa and bC, by=E}];
		\node (EL) [left=0.2cm of E] {};
		\node (ER) [right=0.2cm of E] {};

		\path [blackline] (EL) edge [] (ER);
		\path [blackline] (b) edge [bend right=30] (E);
		\end{tikzpicture}
	\caption{A possible query for \emph{graph simulation} (on the left) and its translation in a guest for loose graph simulations (on the right).}\label{fig:graph_simulation}
\end{figure}
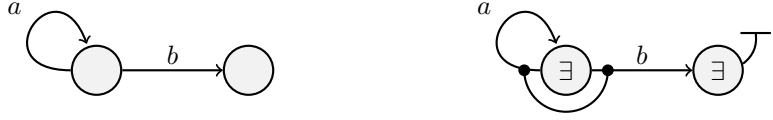

\begin{example}
\cref{fig:graph_simulation} shows a query for GSs and the equivalent guest for LGSs. As seen in \cref{sec:SubIso_via_LS}, the translation preserve the structure of the graph.
\end{example}

\section{Regular languages pattern matching}
\label{sec:RE_via_LS}
Regular languages defines finite sequences of characters (called \emph{words} or \emph{strings}) from a finite alphabet $\Sigma$ \cite{DBLP:books/daglib/0011126}.
Although widely used in text pattern matching, they are also used in graph pattern matching \cite{DBLP:journals/jacm/BarceloLR14,DBLP:journals/siamcomp/MendelzonW95}.
In this section we will restrict ourselves to $\epsilon$-free regular languages, \ie regular languages without the empty word $\epsilon$ \cite{DBLP:journals/tcs/Ziadi96}.
This restriction is quite common, since the empty word is matched by any text or graph and thus it does not represent a meaningful pattern.

\begin{definition} 
	Let $\Sigma$ be an alphabet. $\emptyset$ is a \emph{$\epsilon$-free regular language}. For each $a \in \Sigma$, $\{a\}$ is a \emph{$\epsilon$-free regular language}.
	If $A$ and $B$ are \emph{$\epsilon$-free regular language}, so are $A\cdot B \triangleq \{ vw \mid v \in P \land w \in Q\}$, $A \mid B \triangleq A \cup B$, and $A^+ \triangleq \bigcup_{n \in \N} A^{n+1}$
\end{definition}
In \cite{DBLP:journals/tcs/Ziadi96} it is shown that every regular language without the \emph{empty} string $\epsilon$ can be expressed with the operations defined for \emph{$\epsilon$-free regular languages}.
We will now introduce the pattern matching problem for non-empty $\epsilon$-free regular languages.
In the following let $H = (\Sigma,V_H,E_H)$ and $\Lang$ be respectively a host and a $\epsilon$-free regular language such that $\Lang \neq \emptyset$.

\begin{definition}
	The emptiness problem for regular language pattern matching (RLPM) consist in checking if there is a path $\rho \in \Path_H$ such that $\sigma(\rho) \in \Lang$.
\end{definition}

To solve this problem using LGSs we will use the equivalence between regular languages and non-deterministic finite automata \cite{DBLP:journals/cacm/Thompson68}.

\begin{definition}
	An NFA is a tuple, $N = (\Sigma,Q,\Delta,q_0,F)$ consisting of an alphabet $\Sigma$, a finite set of states $Q$, an \emph{initial} state $q_0$, a set of \emph{accepting} (or \emph{final}) states $F\subseteq Q$ and a \emph{transition} function $\Delta\colon Q \times \Sigma \to \Pwset(Q)$.
	Let $w = a_0,a_1,\dots,a_n$ be a word in $\Sigma^*$.
	The NFA $N$ accepts $w$ if there is a sequence of states $r_0,r_1,\dots,r_{n+1}$ in $Q$ such that $r_0 = q_0$, $r_{i+1} \in \Delta(r_i,a_i)$ for $i = 0,\dots,n$, and $r_{n+1} \in F$. With $\Lang(N)$ we denote the set of words accepted by $N$, \ie its accepted language.
\end{definition}

\begin{remark}
	Any non-empty regular language without $\epsilon$ can be translated to a non-deterministic finite automaton (NFA) with one initial state (say $q_0^\prime$), one final state (say $f$) and s.t.~$\mathrm{in}(q_0^\prime) = \emptyset$ and 	$\out(f) = \emptyset$. Indeed, for $N = (\Sigma,Q,\Delta,q_0,F)$ any NFA s.t.~$\Lang(N) \neq \emptyset$ and $\varepsilon \notin \Lang(N)$
	define $N^\prime = (\Sigma,Q\cup \{q^\prime_0,f\},\Delta^\prime,q^\prime_0, \{f\})$ where:
	\begin{itemize}
		\item for all $a \in \Sigma$, $\Delta^\prime(q^\prime_0,a) \triangleq \Delta(q_0,a)$ and $\Delta^\prime(f,a) = \emptyset$;
		\item for all $q \in Q$ and $a \in \Sigma$, $\Delta^\prime(q,a) \triangleq \Delta(q,a) \cup \{f \mid F \cap \Delta(q,a) \neq \emptyset\}$.
	\end{itemize}
	By construction $\Lang(N) = \Lang(N^\prime)$, $\mathrm{in}(q^\prime_0) = \emptyset$, and $\out(f) = \emptyset$.
\end{remark}

\begin{proposition}\label{prop:RE_via_LS}
	Let $N = (Q,\Sigma,\Delta,q_0,\{f\})$ be a NFA where the initial state $q_0$ does not have any incoming transitions and the only final state $f$ does not have any outgoing ones. Let $H = (\Sigma,V_H,E_H)$ be a host. Let
	\begin{equation*}
	G = {q_0}_{\{\exists\}} \oplus f_{\{\exists,\varnothing\}} \oplus\bigoplus_{\substack{q \in Q\text{, }a \in \Sigma\text{,}q^\prime \in \Delta(q,a)}}\left(q \xrightarrow{a} q^\prime \right)
	\end{equation*}
	Then, there exists a path $\rho \in \Path_H$ in $H$ s.t.~$\sigma(\rho)$ is accepted by $N$ iff there exists a loose graph simulation of $G$ in $H$, \ie $\SSD{G}{H} \neq \emptyset$.
\end{proposition}

\begin{proof}
	It follows from definition of acceptance that if there is $(e_0,\dots,e_n) \in \Path_H$ such that $\sigma(\rho)$ is accepted by $N$ then, there is a sequence
	\begin{equation*}
(p_0,s(e_0)) \xrightarrow{\sigma(e_0)} (p_1,s(e_1)) \xrightarrow{\sigma(e_1)} \dots
	\xrightarrow{\sigma(e_{n-1})} (p_n,s(e_n))
	\xrightarrow{\sigma(e_n)} (p_{n+1}, t(e_n))
	\end{equation*}
	such that $p_0 = q_0$ and $p_{n+1} = f$; for all $i \in \{1,\dots,n\}$ $t(e_{i-1}) = s(e_i)$; for all $i \in \{0,\dots,n\}$ $p_{i+1} \in \Delta(p_i)$.
	Regard the sequence as a graph, say $S$, then $S \in \SSD{G}{H}$ since $S$ is a subgraph of $G \times H$ and $G$ is constructed from $N$ by preserving its transition relation $\Delta$. \Cref{ls:c1,ls:c2,ls:c3} hold since $p_0 = q_0$, $p_n = f$ and $\Unique = \Injective = \emptyset$.
	\Cref{ls:c4} holds since $\{(p_i,\sigma(e_i),p_{i+1})\} \in \Choice(p_i)$ for any $i \in \{0,\dots,n\}$ by construction.
	\Cref{ls:c5} holds since projecting the graph to its first component yields a path from $q_0$ to $f$. Representing $G$ requires space polynomial in the size of $N$.
	Conversely, if there is $S \in \SSD{G}{H}$ then \ref{ls:c5} ensures that there is a path $\rho = (e_0,\dots,e_n)$ in it such that $\pi_1 \circ s(\rho) = q_0$ and $\pi_1 \circ t(\rho) = f$.
	It follows from definition of $E$ that the path $\rho$ is coherent with $\Delta$, \ie $\forall i \in \{0,\dots,n\}\ \pi_1\circ t(e_i) \in \Delta \circ \pi_1\circ s(e_i)$.
	Thus, the the sequence of labels
	\[
		\sigma(\pi_2(\rho)) = ((\pi_2 \circ s(e_0), \sigma(e_0), \pi_2 \circ t(e_0)),\dots,(\pi_2 \circ s(e_n), \sigma(e_n), \pi_2 \circ t(e_n))),
	\]
	 in the second projection of $\rho$ is such that $\sigma(\pi_2(\rho))$ is accepted by $N$.
\end{proof}

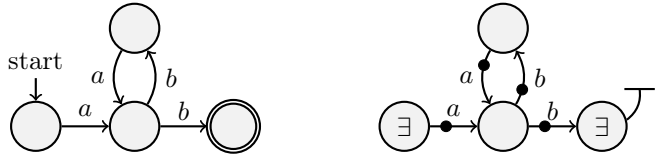
\begin{figure}[t!]
	\centering
		\begin{tikzpicture}[node distance=1.3cm]
		\node [ssplace] (a1) [] {};
		\node [ssplace] (b1) [right of=a1] {};
		\node [ssplace] (c1) [above of=b1] {};
		\node [ssplace,accepting] (f1) [right of=b1] {};

		\node (start) [above=0.3cm of a1] {start};

		\path [blackarrow]
		(a1) edge [->] node[above=-3pt] {$a$} (b1)
		(b1) edge [bend right, ->] node[right=-3pt] {$b$} (c1)
		(c1) edge [bend right, ->] node[left=-3pt] {$a$} (b1)
		(b1) edge [->] node[above=-3pt] {$b$} (f1)
		(start) edge [->] (a1);


		\node [ssplace] (a) [right=1.56cm of f1] {\MustSym};
		\node [ssplace] (b) [right of=a] {};
		\node [ssplace] (c) [above of=b] {};
		\node [ssplace] (f) [right of=b] {\MustSym};

		\path [name path=aC] (a) circle (0.55cm);
		\path [name path=bC] (b) circle (0.55cm);
		\path [name path=cC] (c) circle (0.55cm);
		\path [name path=fC] (f) circle (0.7cm);

		\draw [blackarrow, name path=ab]
		(a) -- node[above=-3pt] {$a$} (b);
		\draw [blackarrow]
		(b) edge [name path=bc, bend right] node[right=-3pt] {$b$} (c);
		\draw [blackarrow]
		(c) edge [name path=cb, bend right] node[left=-3pt] {$a$} (b);
		\draw [blackarrow, name path=bf]
		(b) -- node[above=-3pt] {$b$} (f);

		\path [name intersections={of=ab and aC,by=Ab}];
		\draw [blackdot] (Ab) circle (2pt);

		\path [name intersections={of=bc and bC,by=Bc}];
		\draw [blackdot] (Bc) circle (2pt);

		\path [name intersections={of=cb and cC,by=Cc}];
		\draw [blackdot] (Cc) circle (2pt);

		\path [name intersections={of=bf and bC,by=Bf}];
		\draw [blackdot] (Bf) circle (2pt);

		\node (e) [above right of=f] {};
		\path [name path=pa] (f) -- (e);

		\path [name intersections={of=pa and fC, by=E}];
		\node (EL) [left=0.2cm of E] {};
		\node (ER) [right=0.2cm of E] {};

		\path [blackline] (EL) edge [] (ER);
		\path [blackline] (f) edge [bend right=30] (E);
		\end{tikzpicture}
	\caption{A query for \emph{regular languages} represented as an NFA (left) and as a LGS guest (on the right). The accepted language is $(ab)^+$.}\label{fig:re_ls}
\end{figure}

\begin{example}
 \Cref{fig:re_ls} shows a NFA and a guest identifying the same language.
 These two objects have the same structure (states/nodes and transition/edges).
\end{example}

\section{Subgraph isomorphism with regular path expressions}
\label{sec:REGraphs}

Many approaches found in literature define hybrid notions of similarities,  ``merging'' classical ones such as GS, SGI and RLPM \cite{DBLP:journals/jacm/BarceloLR14,DBLP:journals/fcsc/FanLMTW12}.
These and similar merges are naturally handled by the modular definition of LGS guests. As an example, we discuss \emph{subgraph isomorphism with regular languages} (RL-SGI) \cite{DBLP:journals/jacm/BarceloLR14}.

\begin{definition}\label{def:REgraphs}
	Let $\Sigma$ be a finite alphabet.
	A \emph{graph decorated with regular languages (over $\Sigma$)}  is a tuple $(\Sigma,V,E,\Lang)$ consisting of a set $V$ of nodes, a set $E \subseteq V \times V$ of edges and a labelling function $\Lang : E \to RE_\Sigma$ decorating each edge with a non empty $\epsilon$-free regular language over $\Sigma$.
\end{definition}

\begin{definition}[RL-SGI]\label{def:RE_SGISO}
	Let $H{=}(\Sigma,V_H, E_H)$ be a host and $Q{=}(\Sigma,V_Q, E_Q, \Lang)$ a graph decorated with regular languages.
	We say that there is a \emph{regular-language subgraph isomorphism of $Q$ into $H$}
	 iff there is a pair of injections $\phi : V_Q \injarrow V_H$ and $\eta : E_Q \injarrow \Path_H$ s.t.~for each $e \in E_Q$
	$\phi \circ s(e) = s \circ \eta(e)$,
	$\phi \circ t(e) = t \circ \eta(e)$, and
	$\sigma \circ \eta(e) \in \Lang(e)$.
	Vertexes of paths in $\eta(E_Q)$ cannot appear in $\phi(V_Q)$ except for their source and target, \ie: $
	\forall (e_0,\dots,e_n) \in \eta(E_Q)\ \forall i \in \{1,\dots,n\}\ s(e_i) \not\in \phi(V_Q)$.
\end{definition}
RL-SGI can be seen as a hybrid notion between subgraph isomorphism and RLPM. We will now show how to solve this problem with loose graph simulations by defining a proper translation from its queries to guests.

\begin{proposition}\label{prop:RE_SGISO_via_LS} Let $Q = (\Sigma,V_Q, E_Q, \Lang)$ be a query for RL-SGI. Let
	\begin{equation*}\textstyle
		G = \bigoplus_{v \in V_Q} v_{\{\MustSym\UniqueSym\InjectiveSym\}} \oplus \bigotimes_{e \in E_Q}G_e[q_e/s(e)][f_e/t(e)]
	\end{equation*}
such that $G_e$ is the translation of the automaton $N_e = (\Sigma,V_e,\delta_e,q_e,\{f_e\})$ for $\Lang(e)$, as per \cref{prop:RE_via_LS} and where $q_e$ and $f_e$ are merged if $s(e) = t(e)$.
For each host ${H = (V_H, E_H)}$ there exists a RL-SGI of $Q$ into $H$ iff ${\SSD{G}{H} \neq \emptyset}$.
\end{proposition}
\begin{proof}
	It follows from definition of $G$ that:
	\begin{enumerate*}[label={\emph{(\roman*)}},]
		\item $V_Q$ is a subset of the vertices of $V_G$ and $\Must = \Unique = \Injective = V_Q$;
		\item for any $v \in V_Q$, any $\gamma \in \Choice(v)$, and any $e \in \out(v)$ of $Q$, there is exactly one edge in $\gamma$ that is induced by a transition in $N_e$
	\end{enumerate*}
	Similarly to the proof of \cref{prop:SubIso_via_LS}, \cref{ls:c1,ls:c2,ls:c3} together with the first property ensure that each LGS over $G$ corresponds to an injection w.r.t~$V_Q$.
	It follows from the second property, \cref{prop:RE_via_LS}, \cref{ls:c4,ls:c5} that  every LGS over $G$ contains, for each $e \in E_Q$ a path whose labels, starting and ending nodes lie in $\Lang(e)$ and $V_Q \times V_H$, whereas all other vertices are in $(V_G \setminus V_Q) \times V_H$. Then, $\SSD{G}{H} \neq \emptyset$ iff there are RL-SGIs of $Q$ into $H$.
\end{proof}

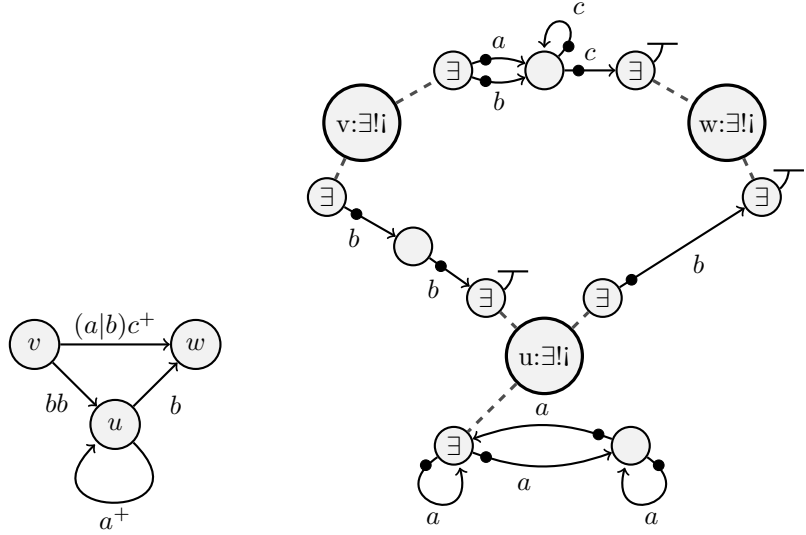
\begin{figure}[!t]
		\centering
			\begin{tikzpicture}[node distance=1.5cm]
			\node [ssplace] (a) [] {$u$};
			\node [ssplace] (b) [above left of=a] {$v$};
			\node [ssplace] (c) [above right of=a] {$w$};

			\path [blackarrow]
			(b) edge [->] node[below left=-3pt] {$bb$} (a)
			(b) edge [->] node[above=-2pt] {$(a|b)c^+$} (c)
			(a) edge [->] node[below right=-2pt] {$b$} (c);

			\draw [blackarrow] (a) edge [in=225, out=-45, loop] node[below=-3pt](uau) {$a^+$} (a);
			\end{tikzpicture}
$\qquad$
			\begin{tikzpicture}[node distance=1.2cm]

			\node [supplace] (b) [] {v:\MustSym\UniqueSym\InjectiveSym};

			\node [minssplace, position=30:{6mm} from b] (q1) {$\exists$};
			\node [minssplace] (i1) [right of=q1] {};
			\node [minssplace] (f1) [right of=i1] {$\exists$};

			\node [supplace, position=-30:{6mm} from f1] (c) {w:\MustSym\UniqueSym\InjectiveSym};

			\node [minssplace, position=-115:{3mm} from b] (q2) {$\exists$};
			\node [minssplace, position=-30:{7.8mm} from q2] (i2) {};

			\node [supplace] (a) [below=3cm of i1] {u:\MustSym\UniqueSym\InjectiveSym};

			\node [minssplace, position=135:{3mm} from a] (f2) {$\exists$};

			\node [minssplace, position=-65:{3mm} from c] (f3) {$\exists$};
			\node [minssplace, position=45:{3mm} from a] (q3) {$\exists$};

			\node [minssplace, position=-135:{9mm} from a] (q4) {$\exists$};
			\node [minssplace, position=0:{18mm} from q4] (f4) {};

			\path [blackarrow]
			(q1) edge [->,bend left=20, name path=q1i11] node[above=-3pt] {$a$} (i1)
			(q1) edge [->,bend right=20, name path=q1i12] node[below=-2pt] {$b$} (i1)
			(i1) edge [->, name path=i1f1] node[above=-3pt] {$c$} (f1);

			\draw [blackarrow] (i1) edge [in=90, out=45, loop] node[above right=-5pt](uau) {$c$} (i1);
			\node [position=45:{10mm} from i1] (i1h) {};
			\path[name path=i1i1] (i1) -- (i1h) {};

			\path [blackarrow]
			(q2) edge [->, name path=q2i2] node[below left=-4pt] {$b$} (i2)
			(i2) edge [->, name path=i2f2] node[below left=-4pt] {$b$} (f2);

			\path [blackarrow]
			(q3) edge [->, name path=q3f3] node[below right=-4pt] {$b$} (f3);

			\path [blackarrow]
			(q4) edge [->,bend right=20, name path=q4f4] node[below left=-3pt] {$a$} (f4)
			(f4) edge [->, bend right=20, name path=f4q4] node[above=-3pt] {$a$} (q4);

			\draw [blackarrow] (q4) edge [in=-75, out=215, loop] node[below=-3pt](uau) {$a$} (q4);
			\node [position=215:{10mm} from q4] (q4h) {};
			\path[name path=q4q4] (q4) -- (q4h) {};

			\draw [blackarrow] (f4) edge [in=-105, out=-35, loop] node[below=-3pt](uau) {$a$} (f4);
			\node [position=-35:{10mm} from f4] (f4h) {};
			\path[name path=f4f4] (f4) -- (f4h) {};

			\draw[simularrow]
			(b) -- (q1)
			(b) -- (q2)
			(f2) -- (a)
			(f1) -- (c)
			(a) -- (q3)
			(a) -- (q4)
			(c) -- (f3);

			\path [name path=q1C] (q1) circle (0.45cm);
			\path [name intersections={of=q1C and q1i11,by=int1}];
			\draw [blackdot] (int1) circle (2pt);

			\path [name intersections={of=q1C and q1i12,by=int2}];
			\draw [blackdot] (int2) circle (2pt);

			\path [name path=i1C] (i1) circle (0.45cm);
			\path [name intersections={of=i1C and i1f1,by=int3}];
			\draw [blackdot] (int3) circle (2pt);

			\path [name intersections={of=i1C and i1i1,by=int4}];
			\draw [blackdot] (int4) circle (2pt);

			\path [name path=q2C] (q2) circle (0.45cm);
			\path [name intersections={of=q2C and q2i2,by=int5}];
			\draw [blackdot] (int5) circle (2pt);

			\path [name path=i2C] (i2) circle (0.45cm);
			\path [name intersections={of=i2C and i2f2,by=int6}];
			\draw [blackdot] (int6) circle (2pt);

			\path [name path=q3C] (q3) circle (0.45cm);
			\path [name intersections={of=q3C and q3f3,by=int7}];
			\draw [blackdot] (int7) circle (2pt);

			\path [name path=q4C] (q4) circle (0.45cm);
			\path [name intersections={of=q4C and q4f4,by=int8}];
			\draw [blackdot] (int8) circle (2pt);

			\path [name intersections={of=q4C and q4q4,by=int9}];
			\draw [blackdot] (int9) circle (2pt);

			\path [name path=f4C] (f4) circle (0.45cm);
			\path [name intersections={of=f4C and f4q4,by=int10}];
			\draw [blackdot] (int10) circle (2pt);

			\path [name intersections={of=f4C and f4f4,by=int11}];
			\draw [blackdot] (int11) circle (2pt);

			\path [name path=bC] (b) circle (0.8cm);
			\path [name path=cC] (c) circle (0.8cm);
			\path [name path=aC] (a) circle (0.8cm);

			\path [name path=f1C] (f1) circle (0.5cm);
			\path [name path=f2C] (f2) circle (0.5cm);
			\path [name path=f3C] (f3) circle (0.5cm);
			\path [name path=f4C] (f4) circle (0.5cm);

			\node (tf1) [above  right of=f1] {};
			\path [name path=rtf1] (tf1) -- (f1);

			\path [name intersections={of=rtf1 and f1C, by=E4}];
			\node (EL4) [left=0.2cm of E4] {};
			\node (ER4) [right=0.2cm of E4] {};

			\path [blackline] (EL4) edge [] (ER4);
			\path [blackline] (f1) edge [bend right=20] (E4);

			\node (tf2) [above  right of=f2] {};
			\path [name path=rtf2] (tf2) -- (f2);

			\path [name intersections={of=rtf2 and f2C, by=E5}];
			\node (EL5) [left=0.2cm of E5] {};
			\node (ER5) [right=0.2cm of E5] {};

			\path [blackline] (EL5) edge [] (ER5);
			\path [blackline] (f2) edge [bend right=20] (E5);

			\node (tf3) [above  right of=f3] {};
			\path [name path=rtf3] (tf3) -- (f3);

			\path [name intersections={of=rtf3 and f3C, by=E6}];
			\node (EL6) [left=0.2cm of E6] {};
			\node (ER6) [right=0.2cm of E6] {};

			\path [blackline] (EL6) edge [] (ER6);
			\path [blackline] (f3) edge [bend right=20] (E6);

			\end{tikzpicture}
			\caption{A RE-SGISO query (left) and simple guests required to encode it (right). Vertices with the same name are highlighted by dashed edges between them.}
			\label{fig:RE_SGISO_1}
\end{figure}

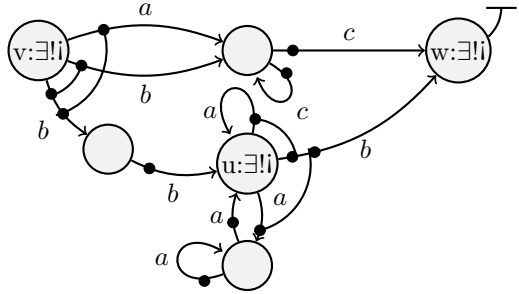
\begin{figure}[!t]
			\centering
			\begin{tikzpicture}[node distance=1.5cm]
			\node [ssplace] (b) {v:\MustSym\UniqueSym\InjectiveSym};
			\node [ssplace] (i1) [right=2cm of b] {};
			\node [ssplace] (c) [right=2cm of i1] {w:\MustSym\UniqueSym\InjectiveSym};

			\node [ssplace] (a) [below of=i1] {u:\MustSym\UniqueSym\InjectiveSym};
			\node [ssplace] (i4) [below=0.6cm of a] {};

			\node [ssplace, position=-55:{8.5mm} from b] (i2) {};

			\path [blackarrow]
			(b) edge [->,bend left=20, name path=bi11] node[above=-3pt] {$a$} (i1)
			(b) edge [->,bend right=20, name path=bi12] node[below=-2pt] {$b$} (i1)
			(i1) edge [->, name path=i1c] node[above=-3pt] {$c$} (c);

			\draw [blackarrow] (i1) edge [in=-70, out=-30, loop] node[below right=-5pt](uau) {$c$} (i1);
			\node [position=-30:{10mm} from i1] (i1h) {};
			\path[name path=i1i1] (i1) -- (i1h) {};

			\path [blackarrow]
			(b) edge [->,bend right=20, name path=bi2] node[below left=-3pt] {$b$} (i2)
			(i2) edge [->,bend right=20, name path=i2a] node[below=-2pt] {$b$} (a);

			\path [blackarrow]
			(a) edge [->,bend right=20, name path=ac] node[below=-3pt] {$b$} (c);

			\path [blackarrow]
			(a) edge [->,bend left=20, name path=ai4] node[above right=-3pt] {$a$} (i4)
			(i4) edge [->,bend left=20, name path=i4a] node[left=-3pt] {$a$} (a);

			\draw [blackarrow] (a) edge [in=120, out=80, loop] node[below left](uau) {$a$} (a);
			\node [position=80:{10mm} from a] (ah) {};
			\path[name path=aa] (a) -- (ah) {};

			\draw [blackarrow] (i4) edge [in=150, out=200, loop] node[left=-3pt](uau) {$a$} (i4);
			\node [position=200:{10mm} from i4] (i4h) {};
			\path[name path=i4i4] (i4) -- (i4h) {};

			\path [name path=bC1] (b) circle (0.9cm);
			\path [name path=bC2] (b) circle (0.6cm);
			\path [name path=cC] (c) circle (0.8cm);
			\path [name path=aC1] (a) circle (0.9cm);
			\path [name path=aC2] (a) circle (0.6cm);

			\node (tc) [above  right of=c] {};
			\path [name path=rtc] (tc) -- (c);

			\path [name intersections={of=rtc and cC, by=E}];
			\node (EL) [left=0.2cm of E] {};
			\node (ER) [right=0.2cm of E] {};

			\path [blackline] (EL) edge [] (ER);
			\path [blackline] (c) edge [bend right=20] (E);

			\path [name intersections={of=bC1 and bi11,by=int}];
			\draw [blackdot] (int) circle (2pt);

			\path [name intersections={of=bC2 and bi12,by=int2}];
			\draw [blackdot] (int2) circle (2pt);

			\path [name intersections={of=bC2 and bi2,by=int3}];
			\draw [blackdot] (int3) circle (2pt);

			\path [name intersections={of=bC1 and bi2,by=int4}];
			\draw [blackdot] (int4) circle (2pt);

			\draw[blackline] ([shift=(int)]0,0) arc (17:-75:0.9cm);
			\draw[blackline] ([shift=(int2)]0,0) arc (-17:-75:0.6cm);

			\path [name intersections={of=aC2 and aa,by=int5}];
			\draw [blackdot] (int5) circle (2pt);

			\path [name intersections={of=aC1 and ai4,by=int6}];
			\draw [blackdot] (int6) circle (2pt);

			\path [name intersections={of=aC1 and ac,by=int7}];
			\draw [blackdot] (int7) circle (2pt);

			\path [name intersections={of=aC2 and ac,by=int8}];
			\draw [blackdot] (int8) circle (2pt);

			\draw[blackline] ([shift=(int5)]0,0) arc (85:10:0.6cm);
			\draw[blackline] ([shift=(int6)]0,0) arc (-75:14:0.9cm);

			\path [name path=i1C] (i1) circle (0.6cm);
			\path [name intersections={of=i1C and i1c,by=int9}];
			\draw [blackdot] (int9) circle (2pt);
			\path [name intersections={of=i1C and i1i1,by=int10}];
			\draw [blackdot] (int10) circle (2pt);

			\path [name path=i2C] (i2) circle (0.6cm);
			\path [name intersections={of=i2C and i2a,by=int11}];
			\draw [blackdot] (int11) circle (2pt);

			\path [name path=i4C] (i4) circle (0.6cm);
			\path [name intersections={of=i4C and i4a,by=int12}];
			\draw [blackdot] (int12) circle (2pt);

			\path [name intersections={of=i4C and i4i4,by=int13}];
			\draw [blackdot] (int13) circle (2pt);
			\end{tikzpicture}
			\caption{A guest obtained via \emph{multiplication} and \emph{addition} operator from the guest in \cref{fig:RE_SGISO_1} on the right and equivalent to the RE-SGISO query in \cref{fig:RE_SGISO_1} on the left.}
		\label{fig:RE_SGISO_2}
\end{figure}

\begin{example}
	\cref{fig:RE_SGISO_1,fig:RE_SGISO_2} show a query for RL-SGI and its translation as a LGS guest.
	As illustrated by \cref{prop:RE_SGISO_via_LS,fig:RE_SGISO_2}, translations are obtained modularly: following \cref{sec:SubIso_via_LS,sec:RE_via_LS}, the first step is to represent nodes and edges of a RL-SGI query in the guests for the SGI and RLPM queries, respectively; the second is to compose them via the guest algebra.
\end{example}

\section{A polynomial fragment of LGSs}\label{sec:POLY}

RLPM and GS are two well-known problems for graph pattern matching and they both admit polynomial time algorithms. Since the emptiness problem for LGSs is NP-complete, we are interested in studying fragments of LGSs that are solvable in polynomial time yet expressive enough to capture the RLPM and GS problems.
The class of simulation problems for guests whose unique and exclusive sets are empty enjoys this property.

Fix $G = (\Sigma_G, V_G, E_G, \Must, \Unique, \Injective, \Choice)$ and $H = (\Sigma_H, V_H,E_H)$. If $\Unique$ and $\Injective$ are empty then, LGSs for $G$ and $H$ are closes under unions hence the union $\bigcup \SSD{G}{H}$ of all LGSs correspond to the greatest LGS. Observe that greatest LGSs may not exist in the general case.
\begin{proposition}\label{prop:max_sim}
	Let $G$ be a guest such that $\Unique = \Injective = \emptyset$. Then $\bigcup \SSD{G}{H}$ is a LGS.
\end{proposition}

\Cref{algo:max_lgs} shows an algorithm for computing the greatest LGS provided that $\Unique$ and $\Injective$ are empty. The algorithm runs in polynomial time and can be readily adapted to compute the greatest LGSs included in a given subgraph of $G \times H$. It follows that the emptiness problem admits a polynomial procedure.

\begin{figure}[t!]
\begin{algorithm}[H]
	\SetKwRepeat{Do}{do}{while}
	\KwData{A host $H$ and a guest $G$ s.t.~$\Unique = \Injective = \emptyset$}
	\KwResult{$\bigcup \SSD{G}{H}$ if it exists, otherwise \emph{false}.}
	$(\Sigma,V_S,E_S) \gets G \times H$\;
	\Do{$V_S \neq V_{S^\prime}$ or $E_S \neq E_{S^\prime}$}{
		$(\Sigma,V_{S^\prime},E_{S^\prime}) \gets (\Sigma,V_S,E_S)$\;
		\ForEach{$(u,v) \in V_{S^\prime}$}{
			\ForEach{$((u,v),a,(u^\prime,v^\prime)) \in \out((u,v))$}{
				\If{$\nexists \gamma \in \Choice(u)$ s.t.~$(u,a,u^\prime) \in \gamma$ and $\forall (u,b,u^{\prime\prime}) \in \gamma$ $\exists (v,b,v^{\prime\prime}) \in \out(v)$
				$((u,v),b,(u^{\prime\prime},v^{\prime\prime})) \in \out((u,v))$}{
					$E_{S^\prime} \gets {E_{S^\prime} \setminus \{((u,v),a,(u^\prime,v^\prime))\}}$\;
				}
			}
			\If{$(\out((u,v)) = \emptyset$ and $\emptyset \not\in \Choice(u))$ or
			$(\exists m \in \Must$ s.t.~$\Path_G(u,m) \neq \emptyset$ and $\forall v^\prime \in V_H\ \Path_{(\Sigma,V_{S^\prime},E_{S^\prime})}((u,v),(m,v^\prime)) = \emptyset)$}{
				$E_{S^\prime} \gets {E_{S^\prime} \setminus (\out((u,v)) \cup \mathrm{in}((u,v)))}$\;
				$V_{S^\prime} \gets {V_{S^\prime} \setminus \{(u,v)\}}$\;
			}
		}
	}
	\lIf{$\forall m \in \Must\ \exists v \in V_H$ s.t.~$(m,v) \in V_S$}{\Return{$(\Sigma,V_S,E_S)$}
	}\lElse{\Return{\emph{false}}}
\end{algorithm}
\caption{Algorithm for computing the greatest loose graph simulation.}\label{algo:max_lgs}
\end{figure}

\begin{theorem}
	Let $H$ be a host and $G$ be a guest such that  $\Unique = \Injective = \emptyset$. Then, the maximal LGS exists and is computed in polynomial time.
\end{theorem}

\begin{proof}
	The algorithm in \cref{algo:max_lgs} starts by computing $G \times H$ and saving it to $(\Sigma,V_S,E_S)$ (Line~1). Afterwards, the \emph{do-while} loop (Lines~2-11) proceeds removing nodes and edges of $(\Sigma,V_S,E_S)$ that do not satisfy \cref{ls:c4,ls:c5}. Lastly (Lines~12-15), \cref{ls:c1} is checked and, if satisfied, $(\Sigma,V_S,E_S)$ is returned, otherwise there is no greatest LGS and the algorithm terminates returning \emph{false}.
	The algorithm runs in polynomial time, since \cref{ls:c1,ls:c4,ls:c5} can be checked in polynomial time (\cref{theo:LS_check_in_P}) and the loop will be performed at most $|V_S| + |E_S|$ times.
	Conditions at Lines~6 and 8 check that edges and nodes satisfy \cref{ls:c4,ls:c5}. If any of these does not hold, the temporary copy of $(\Sigma,V_S,E_S)$, \ie $(\Sigma,V_{S^\prime},E_{S^\prime})$, is updated removing an edge or a vertex. Thus, $V_S \neq V_{S^\prime}$ or $E_S \neq E_{S^\prime}$ iff $(\Sigma, V_S, E_S)$ does not satisfy \cref{ls:c4,ls:c5}.
	After the \emph{do-while} loop, $(\Sigma, V_S, E_S)$ is a (possibly empty) relation that satisfies \cref{ls:c4,ls:c5}. Thus it remains only to check \cref{ls:c1} and this is done at Line~15: if the check fails there is no greatest LGSs otherwise it is the graph $(\Sigma, V_S, E_S)$ returned by the algorithm. Assume otherwise that there is a LGS $(\Sigma, V_M, E_M)$ s.t.~$V_S \subset V_M$ or $E_S \subset E_M$. Then in $(\Sigma, V_M, E_M)$ there is a node or an edge that satisfies \ref{ls:c4} and \ref{ls:c5} and is in $G\times H \setminus (\Sigma, V_S, E_S)$.
	Since it satisfies \ref{ls:c4} and \ref{ls:c5} it cannot be removed by the loop hence it is in $(\Sigma, V_S, E_S)$ --- a contradiction.
\end{proof}

\section{Conclusions and future work}
\label{sec:CONCLUSIONS}

\looseness=-1
In this paper we have introduced \emph{loose graph simulations}, which are relations between graphs that can be used to check structural properties of labelled hosts. LGSs' guests can be represented using a simple graphical notation, but also compositionally by means of an algebra which is sound and complete.
We have shown formally that computing LGSs is an NP-complete problem, where the NP-hardness is obtained via a reduction of subgraph isomorphism to them. Moreover, we have shown that many other classical notions of graph pattern matching are naturally subsumed by LGSs. Therefore, LGSs offer a simple common ground between multiple well-known notions of graph pattern matching supporting a modular approach to these notions as well as to the development of common techniques.

An algorithm for computing LGSs in a decentralised fashion and inspired to the ``distributed amalgamation'' strategy is introduced in \cite{mansutti:msc-thesis}.
Roughly speaking, the host graph is distributed over processes; each process uses its partial view of the host to compute partial solutions to exchange with its peers. Distributed amalgamation guarantees each solution is eventually found by at least one process.

The same strategy is at the core of distributed algorithms for solving problems such as \emph{bigraphical embeddings} and the distributed execution of bigraphical rewriting systems \cite{DBLP:journals/eceasst/MansuttiMP14,mmp:gcm2014,DBLP:journals/corr/MiculanP14b}.
Bigraphs \cite{DBLP:books/daglib/0022395,mp:tr2013,gm:mfps07} have been proved to be quite effective for modelling, designing and prototyping distributed systems, such as \emph{multi-agent systems} \cite{DBLP:conf/dais/MansuttiMP14}.
This similarity and the ability of LGS to subsume several graph problems suggests to investigate graph rewriting systems where redex occurrences are defined in terms of LGSs.

Another topic for further investigation is how to systematically minimise guests or combine sets of guests into single instances, while preserving the semantics of LGSs. Moreover, following what already done in \cref{sec:POLY}, the complexity of various fragments of LGSs still needs to be addressed, \eg defining a fragment that is \emph{fixed-parameter tractable}.
Results in these directions would have a positive practical impact on applications based on LGSs.

\paragraph{Acknowledgements}
We thank Andrea Corradini for his insightful observations on a preliminary version of this work and for proposing the name ``loose graph simulations''.

\providecommand\noopsort[1]{}


\small
\begin{thebibliography}{10}

\bibitem{bioinflmu-308}
J.~Apostolakis, R.~K{\"o}rner, and J.~Marialke.
\newblock { Embedded subgraph isomorphism and its applications in
  cheminformatics and metabolomics}.
\newblock In {\em GCC},
  2005.

\bibitem{DBLP:journals/jacm/BarceloLR14}
P.~Barcel{\'{o}}, L.~Libkin, and J.~L. Reutter.
\newblock Querying regular graph patterns.
\newblock {\em {ACM}}, 61(1):8:1--8:54, 2014.

\bibitem{DBLP:journals/scp/BloomP95}
B.~Bloom and R.~Paige.
\newblock Transformational design and implementation of a new efficient
  solution to the ready simulation problem.
\newblock {\em SCP}, 24(3):189--220, 1995.

\bibitem{DBLP:journals/bmcbi/BonniciGPSF13}
V.~Bonnici, R.~Giugno, A.~Pulvirenti, D.~E. Shasha, and A.~Ferro.
\newblock A subgraph isomorphism algorithm and its application to biochemical
  data.
\newblock {\em {BMC} Bioinformatics}, 14({S-7}):S13, 2013.

\bibitem{Chakrabarti2006GraphML}
D.~Chakrabarti and C.~Faloutsos.
\newblock Graph mining: Laws, generators, and algorithms.
\newblock {\em ACM}, 38:2, 2006.

\bibitem{DBLP:conf/stoc/Cook71}
S.~A. Cook.
\newblock The complexity of theorem-proving procedures.
\newblock In {\em {STOC}}, pages 151--158. {ACM}, 1971.

\bibitem{DBLP:journals/pami/CordellaFSV04}
L.~P. Cordella, P.~Foggia, C.~Sansone, and M.~Vento.
\newblock A (sub)graph isomorphism algorithm for matching large graphs.
\newblock {\em {IEEE}}, 26(10):1367--1372,
  2004.

\bibitem{DBLP:conf/icdt/Fan12}
W.~Fan.
\newblock Graph pattern matching revised for social network analysis.
\newblock In {\em {ICDT}}, pages 8--21. {ACM}, 2012.

\bibitem{DBLP:journals/fcsc/FanLMTW12}
W.~Fan, J.~Li, S.~Ma, N.~Tang, and Y.~Wu.
\newblock Adding regular expressions to graph reachability and pattern queries.
\newblock {\em FOCS}, 6(3):313--338, 2012.

\bibitem{DBLP:journals/pvldb/FanWWD14}
W.~Fan, X.~Wang, Y.~Wu, and D.~Deng.
\newblock Distributed graph simulation: Impossibility and possibility.
\newblock {\em {PVLDB}}, 7(12):1083--1094, 2014.

\bibitem{DBLP:journals/cacm/Floyd62a}
R.~W. Floyd.
\newblock Algorithm 97: Shortest path.
\newblock {\em {ACM}}, 5(6):345, 1962.

\bibitem{gm:mfps07}
D.~Grohmann and M.~Miculan.
\newblock Directed bigraphs.
\newblock In {\em Proc.~MFPS}, volume 173 of {\em ENTCS}, pages 121--137. Elsevier, 2007.

\bibitem{DBLP:conf/focs/HenzingerHK95}
M.~R. Henzinger, T.~A. Henzinger, and P.~W. Kopke.
\newblock Computing simulations on finite and infinite graphs.
\newblock In {\em {FOCS}}, pages 453--462. {IEEE}, 1995.

\bibitem{DBLP:books/daglib/0011126}
J.~E. Hopcroft, R.~Motwani, and J.~D. Ullman.
\newblock {\em Introduction to automata theory, languages, and computation -
  international edition}.
\newblock Addison-Wesley, 2003.

\bibitem{DBLP:conf/sigcomm/LischkaK09}
J.~Lischka and H.~Karl.
\newblock A virtual network mapping algorithm based on subgraph isomorphism
  detection.
\newblock In {\em {VISA}}, pages 81--88. {ACM}, 2009.

\bibitem{mansutti:msc-thesis}
A.~Mansutti.
\newblock Le simulazioni lasche: definizione, applicazioni e computazione
  distribuita.
\newblock Master's thesis, University of Udine, 2016.

\bibitem{DBLP:journals/eceasst/MansuttiMP14}
A.~Mansutti, M.~Miculan, and M.~Peressotti.
\newblock Distributed execution of bigraphical reactive systems.
\newblock {\em {ECEASST}}, 71, 2014.

\bibitem{DBLP:conf/dais/MansuttiMP14}
A.~Mansutti, M.~Miculan, and M.~Peressotti.
\newblock Multi-agent systems design and prototyping with bigraphical reactive
  systems.
\newblock In K.~Magoutis and P.~R. Pietzuch, editors, {\em Proc.~{DAIS}},
  volume 8460 of {\em LNCS}, pages 201--208. Springer, 2014.

\bibitem{mmp:gcm2014}
A.~Mansutti, M.~Miculan, and M.~Peressotti.
\newblock Towards distributed bigraphical reactive systems.
\newblock In R.~Echahed, A.~Habel, and M.~Mosbah, editors, {\em Proc.~GCM},
  page~45, 2014.

\bibitem{DBLP:journals/siamcomp/MendelzonW95}
A.~O. Mendelzon and P.~T. Wood.
\newblock Finding regular simple paths in graph databases.
\newblock {\em {SIAM}}, 24(6):1235--1258, 1995.

\bibitem{mp:tr2013}
M.~Miculan and M.~Peressotti.
\newblock Bigraphs reloaded: a presheaf presentation.
\newblock Technical Report UDMI/01/2013, Univ.~of Udine, 2013.

\bibitem{DBLP:journals/corr/MiculanP14b}
M.~Miculan and M.~Peressotti.
\newblock A {CSP} implementation of the bigraph embedding problem.
\newblock {\em CoRR}, abs/1412.1042, 2014.

\bibitem{DBLP:books/daglib/0022395}
R.~Milner.
\newblock {\em The Space and Motion of Communicating Agents}.
\newblock CUP, 2009.

\bibitem{DBLP:books/daglib/0002478}
P.~A. Pevzner.
\newblock {\em Computational molecular biology - an algorithmic approach}.
\newblock {MIT} Press, 2000.

\bibitem{DBLP:conf/gg/1997handbook}
G.~Rozenberg, editor.
\newblock {\em Handbook of Graph Grammars and Computing by Graph
  Transformations, Volume 1: Foundations}. World Scientific, 1997.

\bibitem{DBLP:journals/cacm/Thompson68}
K.~Thompson.
\newblock Regular expression search algorithm.
\newblock {\em {ACM}}, 11(6):419--422, 1968.

\bibitem{DBLP:journals/jacm/Ullmann76}
J.~R. Ullmann.
\newblock An algorithm for subgraph isomorphism.
\newblock {\em {ACM}}, 23(1):31--42, 1976.

\bibitem{DBLP:conf/icdm/YanH02}
X.~Yan and J.~Han.
\newblock gspan: Graph-based substructure pattern mining.
\newblock In {\em {ICDM}}, pages 721--724. {IEEE}, 2002.

\bibitem{DBLP:journals/tcs/Ziadi96}
D.~Ziadi.
\newblock Regular expression for a language without empty word.
\newblock {\em TCS}, 163(1{\&}2):309--315, 1996.

\end{thebibliography}
\end{document}